\documentclass[11pt]{article}
\pdfoutput=1

\usepackage[utf8]{inputenc}
\usepackage{amsmath,amsthm,amssymb}
\usepackage[paper=a4]{typearea}
\areaset{418pt}{609pt}
\usepackage[bottom]{footmisc}
\usepackage{paralist}
\usepackage{mdframed}
\usepackage{color}
\definecolor{citecolor}{rgb}{0.8,0,0}
\definecolor{linkcolor}{rgb}{0,0,0.8}
\definecolor{urlcolor}{rgb}{0,0,0.8}
\usepackage[colorlinks=true,linktocpage=true,hyperfootnotes=false]{hyperref}
\usepackage{adjustbox}
\usepackage{tikz}
\usetikzlibrary{arrows.meta}
\usetikzlibrary{shapes,arrows,chains,backgrounds}
\usetikzlibrary{calc}

\hypersetup{%
  citecolor=citecolor,
  linkcolor=linkcolor,
  urlcolor=urlcolor
}
\usepackage{microtype}
\usepackage[mathic]{mathtools}

\makeatletter
\def\th@plain{%
  \thm@notefont{}%
  \itshape %
}
\def\th@definition{%
  \thm@notefont{}%
  \normalfont %
}
\makeatother

\theoremstyle{definition}
\newtheorem{theorem}{Theorem}

\theoremstyle{theorem}
\newtheorem{proposition}[theorem]{Proposition}
\theoremstyle{definition}

\newtheorem{example}[theorem]{Example}

\theoremstyle{theorem}
\newtheorem*{keywords}{Keywords}

\usepackage[T1]{fontenc}
\usepackage{pdfpages}
\usepackage{amsmath}
\usepackage{amsthm}
\usepackage{amssymb}
\usepackage{amsfonts}
\usepackage{natbib}

\usepackage[frozencache=true]{minted}

\newcommand{\code}[1]{\texttt{#1}}

\begin{document}

\title{Bringing Together Dynamic Geometry Software and the Graphics Processing Unit}

\author{Aaron Montag\thanks{Zentrum Mathematik (M10), Technische Universität München, 85747~Garching, Germany.\newline E-mail addresses: \url{montag@ma.tum.de} and \url{richter@ma.tum.de}. The authors were supported by the DFG Collaborative Research Center TRR 109, ``Discretization in Geometry and Dynamics''. This is a preprint.} \and Jürgen Richter-Gebert$^*$.}

\maketitle

\begin{abstract}
We equip dynamic geometry software (DGS) with a user-friendly method that enables massively parallel calculations on the graphics processing unit (GPU). This interplay of DGS and GPU opens up various applications in education and mathematical research. The GPU-aided discovery of mathematical properties, interactive visualizations of algebraic surfaces (raycasting), the mathematical deformation of images and footage in real-time, and computationally demanding numerical simulations of PDEs are examples from the long and versatile list of new domains that our approach makes accessible within a DGS. We ease the development of complex (mathematical) visualizations and provide a rapid-prototyping scheme for general-purpose computations (GPGPU).

The possibility to program both CPU and GPU with the use of only one high-level (scripting) programming language is a crucial aspect of our concept. We embed shader programming seamlessly within a high-level (scripting) programming environment. The aforementioned requires the symbolic process of the transcompilation of a high-level programming language into shader programming language for GPU and, in this article,  we address the challenge of the automatic translation of a high-level programming language to a shader language of the GPU.
To maintain platform independence and the possibility to use our technology on modern devices, we focus on a realization through WebGL.\end{abstract}

\begin{keywords}
dynamic geometry software, GPU, transcompilation, shader programming, script language, modern devices
\end{keywords}

\section{Introduction}\label{introduction}
Graphics processing units (GPUs) open up many new possibilities. However, extensive programming expertise and additional development effort are required to program them. In the first part of this article, we propose a concept to overcome this hurdle by seamlessly embedding shader programming language for the GPU within a high-level (scripting) programming language for rapid software prototyping. The concept includes the formal process of a symbolic transcompilation from a high-level language to a GPU shading language.

\cite{richter2010power} demonstrated that dynamic geometry software systems (DGS) equipped with a programming environment could be used advantageously for many mathematical scenarios. In the second part of this article, we will show how such an environment can be extended to utilize the GPU by the formal methods developed in the first part of the project. It provides the mathematician with a relatively easy-to-use tool for the following tasks, among many others: finding certain mathematical conjectures experimentally, the rapid prototyping of mathematical algorithms in parallel, building several visualizations of, for example, algebraic surfaces, fractals, fluid simulations and particle simulation within the framework of a classical DGS.
This method is available inside the browser and on modern devices such as tablets and mobile phones. Additionally, the field of mathematics education can greatly benefit from the seamless integration of GPU programming into DGS.

\subsection{Technical background}

Computations on modern devices can, in general, be executed either on the CPU (central processing unit) or GPU (graphics processing unit).
While CPUs were designed for sequential general-purpose computations, GPUs were initially invented to accelerate the graphics output of the computer. 
In order to obtain a high degree of flexibility in
designing the graphical operations,
programmable shader units were introduced within the GPU rendering pipeline.

Shaders can be considered as little programs which the GPU can execute massively in
parallel. According to \cite{owens2007survey}, a typical GPU rendering pipeline consists of a
vertex and fragment shader (among possibly other shaders).
Conventionally, the vertex shader performs operations on the vertices of
a three-dimensional mesh, while the fragment shader -- sometimes also
called as pixel shader -- computes the color for every single pixel. Soon,
it was discovered that these programmable shaders open the door for
general-purpose computations on the GPU (GPGPU), that were conventionally performed on the CPU. 
In addition to the generation of images, a shader can also be used whenever the execution of the same program at an independent set of data points is required. Often numerical simulations can be built on such a computational scheme.

With this approach, the gap between CPU and GPU, with regard to \emph{the technical capabilities}, gradually starts closing and a comprehensive set of traditional CPU-targeted tasks can be executed on the GPU. 

For many visualizations and scientific computations, the use of \emph{parallel architectures} (such as the GPU) is essential, because nowadays parallel architectures outperform the computational power of a \emph{single-threaded} CPU by several magnitudes. In the future, one might expect that these differences will become even more crucial because the scale of single-threaded computations is stagnating due to power density issues \citep{sutter2005free}, whereas the parallel architecture of GPUs can still benefit from exponential growth in its number of cores, \cite{nickolls2010gpu}. %

A more modern approach for GPGPU programming is to use NVIDIA CUDA C \citep{nvidia2017nvidia} or OpenCL C/C++ \citep{munshi2011opencl} instead of graphics shaders.
CUDA is a framework provided by NVIDIA that bypasses the use of standard shaders to enable general-purpose parallel computation on the GPU by directly accessing
the GPU. However, it is limited to NVIDIA hardware, and CUDA is not
available on browsers. The Khronos Group introduced OpenCL as an alternative that targets a broader range of hardware for GPGPU computation.

Both of the techniques provide programmable kernels that run close to the hardware for the computations. 
Nevertheless, CUDA and OpenCL are not available on every platform. In particular, these techniques are not available within browsers. The endeavors to make OpenCL available on the browser
through WebCL have declined: for instance, Mozilla (announced by \cite{bugzilla}) decided not to implement WebCL in favor of WebGL compute shaders.

This article will be limited to to \emph{shader}-based approaches, although many results could be easily transferred for target architectures such as CUDA or OpenGL. 
The motivation for our particular focus on the shader-based approach is \emph{WebGL} \citep{marrin2011webgl}, which can run shaders on recent web browsers without additional plugins and has extensive \emph{cross-platform compatibility}. WebGL, in particular, can be used for demonstration purposes within a webpage and it is suitable for many modern devices as a target platform \citep{evans20143d}.

\subsection{The gap in programming concepts between CPU and GPU }

Even tough tasks that were initially done on the CPU can now be
accelerated on the GPU, there is still a significant gap in \emph{programming concepts} between CPU and GPU.
This clear distinction may be desired if the aim is to create most-efficient software that exploits the maximum possible computational power tailored for the available hardware architecture.
The standard application programming interfaces (APIs) such as OpenGL, DirectX (and also OpenCL C/C++ and CUDA C), are made for programmers who want to make very conscious hardware decisions.
However, if a programmer wants to benefit from the advantages of the GPU within an abstract
scripting language (for instance for rapid software prototyping), then the
development efforts should also be minimized.
For scientific computing, the time-saving advantages of the GPUs are often not used because 
the gain in performance time often does not justify the additional development effort and time \citep{pycuda}.

In our opinion, there are several difficulties in integrating shader programs (and also kernels for CUDA or OpenCL) within other programs that keep programmers from using the GPU:

\begin{description}
\item[Separation.]
  The shaders are separated from the main program. In many APIs, a shader program has to be provided through an external file or a string containing the shader source, and thereon the source is compiled with a special compiler for the GPU. This separation in code increases the complexity of a program.
\item[Lines of code.]
  A lot of technical boilerplate code  (i.e. code that is duplicated verbatim for various applications with only slight modifications) has to be written to create even simple applications.
\item[Own language.]
  Shaders are written in a programming language that has been designed for graphics purposes and will be compiled for the shaders for the GPU. Even to map other non-graphical computations to the GPU, the programmer has to become acquainted with this programming language.
  Furthermore, there is a semantic discrepancy when shaders that are written in a compiled language are used within a scripting language. A scripting language is, in contrast, designed to be interpreted at running time. 
\item[Platform-dependence.]
  The use of shaders poses a limitation in the number of platforms on which the program will run. This often demands additional programming to make the software run on a sufficiently large set of different hardware.
\end{description}

Standardized APIs like OpenGL have overcome the last limitation. With OpenGL, shader code can be compiled for GPUs of different
vendors. Often, these shaders are compiled at run-time on a target machine. WebGL \citep{marrin2011webgl} utilizes this mechanism to use OpenGL within different browsers.

To our knowledge, for most scripting languages, there are only APIs
available that require the programmer to write the shader programs
in a particular low-level language that has been designed for the GPU.
(For example, if OpenGL is accessed within JavaScript through WebGL, the programmer still has to provide shader codes written in the OpenGL
Shading Language (\emph{GLSL})).

\subsection{Our objectives for a high-level language with GPU support}
The project aims to develop a concept that makes it easier to embed GPU shaders in another high-level programming language (referred to as \emph{host language} in the following). Our main aims are:

\begin{description}
\item[Smooth integration.] The code for the GPU shaders should be smoothly integrated within the
host language with only a minimum of additional language
  constructs.
  Ideally, the programmer should not need to indicate (or even be aware of) whether he or she wants to use the graphics card. In particular, the programmer, who does not program close to the shader, does not have to indicate the splitting point between CPU and GPU code. With this approach,
  a large class of user-defined functions should be usable both on CPU and GPU, and the amount of written code would be drastically reduced.
\item[Efficiency.]
  At the same time, the performance at running time should remain comparable with applications that traditionally use graphics shaders.
\item[Portability \& Stability.]
  Ideally, the programs should run on many platforms. For a visualization framework with a broad outreach, we propose using WebGL.
\item[Versatile applications.]
   Our concept is intended to provide a tool for both real-time rendering of images and basic GPGPU programming.
\item[Compilation of untyped language.] The scripting language will be transcompiled to a shader language at running time. (The shader language can then be compiled to the vendor-specific binary code for the CPU.) 
   As the host language is assumed to be a dynamically
  typed scripting language, the programmer should not decide upon
  GPU-specific types to be chosen. In our approach, the types are
  \emph{inferred automatically}. The types should be as weak as possible in order to make the program run fast. In fact, this requirement poses strong formal constraints on the software architecture of the transpiler. Major parts of this article are dedicated to this issue.

\item[Mathematically-oriented user.] The language should fit the requirements of a mathe\-mat\-ically-oriented user. For example, variables for numbers should be able to store complex-valued numbers, if needed. Vectors of arbitrary length, basic linear algebra operations for matrices and vectors should be part of the language as well.
\end{description}

In Section \ref{related-work}, we present an overview of related work, then in
Section \ref{concept}, we introduce the primary process of the transcompilation. In
Section \ref{example-implementation-cindygl}, we provide a brief outline of our sample implementation \emph{CindyGL} and a set of examples. \emph{CindyGL} is a
JavaScript-based implementation of the concepts introduced in this paper. \emph{CindyGL} can
transcompile the scripting language \emph{CindyScript} to \emph{GLSL} and thus can be
used to create widely portable visualizations and GPGPU programs that
can be run within the browser.

\subsection{Related work}\label{related-work}

In dynamic geometry software, often an expression is evaluated for many data points. An aim addressed by approach is the challenge to render images, where each pixel on the screen gets assigned a color by a given function. \cite{liste2014color} suggests a technique to render curves in Geogebra via a so-called ``sweeping-line'' approach on the CPU, which turns out to be very slow even on recent hardware. Also the dynamic geometry software Cinderella 2 includes a \mintinline{CindyScript}{colorplot}-command to produce similar plots \citep{richter2012cinderella}. However, all these approaches are very slow due to computations on the CPU, and complex images cannot be rendered in real-time. To our best knowledge, no dynamic geometry software can utilize the GPU for such tasks that would require a seamless embedding of GPU programming within a high-level environment.

However, some (non-DGS) projects aim for seamless integration of CUDA or OpenCL within another scripting programming language.
\emph{Copperhead} \citep{copperhead} is probably the most remarkable related project. Copperhead can translate a subset of the scripting
language Python to CUDA at running time. The types are modeled in a minimal Hindley-Milner type system and inferred for the GPU.
Hence Copperhead enables GPGPU
programming within a high-level host language, without requiring the programmer to provide separate code in a low-level language designed for the GPU.

Another approach in integrating GPU accelerated code in another programming language is the introduction of a new datatype for (large) arrays that is permanently stored on the GPU. For example, \emph{Accelerator for C\#} \citep{tarditi2006accelerator} implements a new \texttt{parallel\ array} datatype which one can use for element-wise operations, reductions, affine transformations and linear algebra on the GPU. In a similar spirit, in MATLAB, \texttt{gpuArray} objects were introduced with release R2010b. The \texttt{gpuArray} objects represent matrices that are stored on the GPU, and several operators are overloaded for objects of the \texttt{gpuArray} class.
For a more extensive description on GPU programming in MATLAB, we refer to \cite{reese2012gpu}. Furthermore, a subset of MATLAB code can be executed on the GPU through the \texttt{arrayfun} function if it is applied to a \texttt{gpuArray}. However, the function within \texttt{arrayfun} cannot access variables from the workspace. Matlab utilizes CUDA for its GPU computations. In addition, \cite{pycuda} introduce a class named \texttt{gpuarray} in \emph{PyCUDA} that supports abstractions for many component-wise operations. More complicated functions require the programming of kernels in CUDA C.

Nevertheless, a substantial difference compared to our project is that CUDA is used to access the GPU. In contrast, we use WebGL and focus more on the mathematically-oriented user. Because of our choice of WebGL, running GPGPU code might be slower than approaches that are based on the CUDA architecture, which has been designed for the GPGPU purpose, but we win  in terms of compatibility on every machine that provides a browser with WebGL support and can use a pipeline that has been designed for real-time visualizations.

\section{Concept}\label{concept}
Before going into the details of the symbolic process of type inference, we will describe the general setup that clarifies the roles between author, scripting language, interpreter/compiler, GPU and CPU usage.

We recommend generating the GPU code at run-time, as has been introduced by \cite{pycuda}. The generation of GPU code at run-time increases the portability of the software, and the scheme goes hand in hand with the design idea of a typical scripting language. Furthermore, this provides the possibility for the self-adaption of a program according to the given input data.

 The aim is to enable efficient GPU computations without posing difficulties in the development. Therefore most of the preparation for running binary code on the GPU will be done on the machine of the user. Since the results of the GPU compilation will be cached, the penalties in running time will mostly only occur during the first execution of code that is suitable for parallelization.

We propose the following concept, which is illustrated in Fig.\,\ref{fig: scheme-flow-chart}:
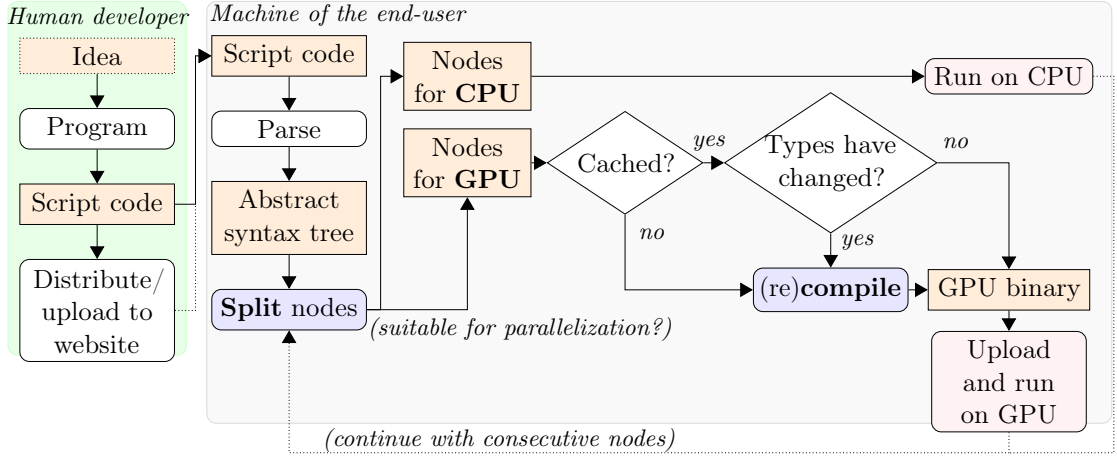
\begin{figure}
  \begin{center}
  \begin{adjustbox}{max width=\textwidth}
\begin{tikzpicture}[%
    >=triangle 60,              %
    start chain=going below,    %
    node distance=5mm and 28mm, %
    every join/.style={norm},   %
    ]
\tikzset{
  base/.style={draw, on chain, on grid, align=center},
  thing/.style={base, rectangle, text width=20mm, fill=orange!15},
  test/.style={base, diamond, aspect=1.5, inner sep=-3pt, text width=24mm, fill=white},
  action/.style={thing, rounded corners, fill=white},
  run/.style={action, rounded corners, fill=red!5},
  norm/.style={->, draw},
  free/.style={->, draw, lcfree},
  cong/.style={->, draw, lccong},
  it/.style={font={\small\itshape}}
}

 \draw[green!20!white, fill=green!10!white, rounded corners] (-1.3,-4.4) rectangle (1.3,.8);
 \draw[gray!30!white, fill=gray!5!white, rounded corners] (1.6,-5.4) rectangle (14.8,.8);
\node [thing, densely dotted] (idea) {Idea};
\node [above=0mm of idea, it] {Human developer};
\node [action, join] {Program};
\node [thing, join] (sc1) {Script code};
\node [action, join] (dist) {Distribute/ upload to website};

\node [thing, right=of idea] (sc) {Script code};

\draw[->]  (sc1.east) -- ++(3mm,0) |- (sc);
\draw[->, densely dotted]  (dist.east) -- ++(3mm,0) |- (sc);

\node [above=0mm of sc, xshift=0.7cm, it] {Machine of the end-user};
\node [action, join] (parser) {Parse};
\node [thing, join] (ast) {Abstract syntax tree};
\node [action, join, fill=blue!10] (split) {\textbf{Split} nodes};

\node [thing, right=of parser, xshift=-.2cm, yshift=-.5cm, text width=16mm] (cgpu) {Nodes for \textbf{GPU}};
\node [test, right=of cgpu, join, xshift=-.5cm] (cached) {Cached?};
\node [test, right=of cached,  xshift=0.2cm] (typechange) {Types have changed?};
\node [action, fill=blue!10] (recomp) {(re)\textbf{compile}};
\node [thing,right=of recomp, join, xshift=-.2cm, text width=21mm] (gpubin) {GPU binary};
\node [run,join, yshift=.2cm] (rungpu) {Upload and run on GPU};

\node [thing, right=of sc, xshift=-.2cm, yshift=-.3cm,  text width=16mm] (ccpu) {Nodes for \textbf{CPU}};

\node [on grid, draw, rectangle, rounded corners, fill=red!5] at (ccpu -| rungpu) (runcpu) {Run on CPU};
\draw[->] (ccpu) -- (runcpu);

\draw[->] (split.east) -- ++(2mm,0) |- (ccpu);

\draw[->]  (split) -| node[black, xshift=2em, yshift=-0.75em, it] {(suitable for parallelization?)}(cgpu);

\draw[->] (cached) |- node [black, very near start, xshift=1em, it] {no} (recomp);
\draw[->] (cached) -- node [black, near start, yshift=0.75em, it] {yes} (typechange);
\draw[->] (typechange) -- node [black, near start, xshift=1em, it] {yes} (recomp);
\draw[->] (typechange) -| node [black, very near start, yshift=0.75em, it] {no} (gpubin);

\draw[->, densely dotted]  (rungpu.south) -- ++(0,-3mm) -| node[black,it, yshift=.5em, xshift=8em] {(continue with consecutive nodes)} (split);
\draw[-, densely dotted]  (runcpu.east) -- ++(3mm,0) |- ($(rungpu.south)+(0,-3mm)$);
\end{tikzpicture}
\end{adjustbox}

    \caption{Our proposed scheme, which extends the concept from \cite{pycuda}, for a seamless embedding of shader code within the environment of a scripting language. For simplicity, data transfers between GPU and CPU are omitted in this figure.}
    \label{fig: scheme-flow-chart}
  \end{center}
\end{figure}
The author of the content will directly program his or her idea
through a single scripting language, regardless whether he or she is aiming for a GPU or CPU computation.
When the end user executes the scripting code, the interpreter will parse the scripting code by standard techniques as described by \cite{levine2009flex} to obtain an abstract syntax tree (referred to as AST in the following). Instead of instantaneously executing the entire AST, the interpreter computes a split of the AST such that those parts that are suitable for parallelization are separated from those parts that are more appropriate for execution on the CPU. On its
first execution, the interpreter compiles the code suitable
for the GPU to a shading language of the GPU and then uses the host's GPU compiler to compile this shading code to GPU binary. For the subsequent calls
of the same code, the compilation is skipped. Only if the types of the
contained variables have changed because of different parameters
(e.g. if a real-valued variable becomes complex valued) a recompilation is
enforced. The GPU binary is uploaded to the GPU and executed on the GPU.

More details about the splitting, type detection, and transcompilation follow in Sections \ref{detection-of-parts-for-parallelization}, \ref{type-detection} and \ref{transcompilation}, respectively.
The lazy synchronization of data between GPU and CPU will be addressed in Section \ref{lazysec}.

\subsection{Detection of parts for
parallelization and splitting the code}\label{detection-of-parts-for-parallelization}

We assume that the host scripting language has a set of specific operations that are suitable for parallel computations on the GPU without any modifications (\cite{herlihy2011art} coined the term ``embarrassingly parallel'' expressions). These can be instructions that belong to the \emph{single instruction, multiple data} (SIMD) class, i.e., very similar instructions that are to be executed on a large set of data points.
A typical operation that corresponds to this scheme is the \code{map} operator of a functional language (if it is applied to a function without side effects and a sufficiently large array). The \code{map} operator applies a given function to each element of an array and returns the results as an array of the same dimensions. Another operation, which was implemented in our sample realization \emph{CindyGL}, is the semantically similar \mintinline{CindyScript}{colorplot} command. The command generates a texture, i.e. a large 2-dimensional array of pixels, where the color of each pixel is computed by a given function taking the pixel coordinate as an input.

Within this function, user-defined functions should also be callable. However, it can not be expected that a transcompilation can work for all functions. For instance, the \emph{GLSL} shading language 1.0 \citep{simpson2009opengl} does not allow any recursion.

An alternative language construct that is also suitable for our concept is the introduction of a special array type that is permanently stored on the GPU. This is familiar from the parallel array in Accelerator for C\# \citep{tarditi2006accelerator}, \texttt{gpuarray} in PyCUDA \citep{pycuda}, and \texttt{gpuArray} in Matlab \citep{reese2012gpu}. The point-wise operations on this array are suitable for parallelization.

All these constructs have at least one \emph{running variable} that takes a different value for each call of the function. Within \code{map}, a function is applied to this variable, or the pixel coordinate within \mintinline{CindyScript}{colorplot} becomes the running variable. Those terms that are invariant on the running variables should be computed only once instead of being massively parallelized. To determine them, we build a directed \emph{dependency graph} $G = (V, E)$ that
contains all variables and all expressions that appear within the function as nodes. Without loss of generality, we assume that there is only one running variable. Otherwise, we can assume the running variables are stored together in a single array. Let $v_0$ denote the running variable.
There is a directed edge $(a,b)\in E$ iff the variable/term $a$ depends on $b$. This might be because $a$ is a variable and is assigned to $b$ or $a$ is an expression that that contains $b$ and hence its value is dependent on $b$ (i.e., $b$ is a child of $a$ in the AST).
Furthermore, all variables that are modified within a conditional loop are dependent on the condition of the loop (e.g.,~the
boolean variable within an if-clause, or the number of repetitions of a
loop). %

Now a set $D \subset V$ of terms that depend on the running variable is computed: $D \subset V$ contains precisely those nodes that can be reached in $G$ from $v_0 \in V$ and can be determined by a depth-first search in $G$.
The node $r\in V$ corresponding to the result term either is contained in $D$ or is not contained in $D$. If $r\not\in D$, then the computation is independent on the running variable and thus always takes the same value and can be computed once on the CPU. If $r\in D$, all terms in $D$  will be marked for a transcompilation to the GPU.

\newcommand{\E}{\mathcal{E}}
All nodes in
$v \in V \setminus D$ that have an immediate successor in $D$ (i.e., there is an
$a\in D$ such that $(v,a)\in E$) will be marked as uniform variables and form a set $U$.

Let $\E=D \cup U$ denote the set of relevant expressions.
The following transformations will be only applied to $\E$. The values of the terms in $U$ will be computed once on the CPU and used as input parameter for the GPU program.

This enables us to obtain an almost optimal split of code in CPU and GPU computations.

\begin{example}
Consider the expression \mintinline{CindyScript}{1/2+1/2*sin(|#|-seconds())} \footnote{This is an expression to generate the colorplot of a centered sinusoidal wave, see  Example \ref{waveexample}.}. Let \mintinline{CindyScript}{#} be the running variable $v_0$. The expression generates a dependency graph $G$ shown in Fig. \ref{fig:tree}. The three terms \mintinline{CindyScript}{seconds()}, and twice \mintinline{CindyScript}{1/2} are the uniform expressions $U$, i.e. they are independent from $v_0 = \mintinline{CindyScript}{#}$ and are computed once on the CPU because they attain their value independently from $v_0$.
\end{example}

\begin{figure}
  \begin{center}
    \begin{adjustbox}{max width=\textwidth}
    \tikzset{
  base/.style={draw, on chain, on grid, align=center},
  thing/.style={base, rectangle, text width=20mm, fill=orange!15},
  test/.style={base, diamond, aspect=1.5, inner sep=-3pt, text width=24mm, fill=white},
  action/.style={thing, rounded corners, fill=white},
  run/.style={action, rounded corners, fill=red!5},
  treenode/.style = {shape=rectangle,
                     draw, align=center},
  root/.style     = {treenode, font=\Large, bottom color=red!30},
  D/.style     = {fill=orange!15},
  U/.style     = {fill=blue!10},
  env/.style      = {treenode, font=\ttfamily\normalsize},
  dummy/.style    = {circle,draw}
}
\begin{tikzpicture}
  [
   >=triangle 60,     
    grow                    = right,
    sibling distance        = 4em,
    level distance          = 8em,
    edge from parent/.append style={<-},
    sloped
  ]
  \node [D] {\code{1/2+1/2*sin(|\#|-seconds())}}
    child { node [U] {\code{1/2}}
        child { node [xshift=-1.5cm, above] {\code{2}}}
        child { node [xshift=-1.5cm,below] {\code{1}}}
    }
    child {
      node [D] {\code{1/2*sin(|\#|-seconds())}}
      child { node [U] [above] {\code{1/2}}
        child { node [xshift=-1.5cm,above] {\code{2}}}
        child { node [xshift=-1.5cm,below] {\code{1}}}
      }
      child { node [D] {\code{sin(|\#|-seconds())}}
        child { node [D, xshift=0.6cm] {\code{|\#|-seconds()}}
          child { node [U] {\code{seconds()}} }
          child { node [D, xshift=-1cm] {\code{|\#|}}
            child { node [D, xshift=-1.5cm] {\code{\#}} node[xshift=-0.9cm]{$= v_0$}}
        }
        }
      }
    };
\end{tikzpicture}
    \end{adjustbox}
    \caption{Dependency graph $G$ for expression \code{1/2+1/2*sin(|\#|-seconds())}. Nodes in $D$ are highlighted in orange, nodes in $U$ are highlighted in blue.}
    \label{fig:tree}
  \end{center}
\end{figure}
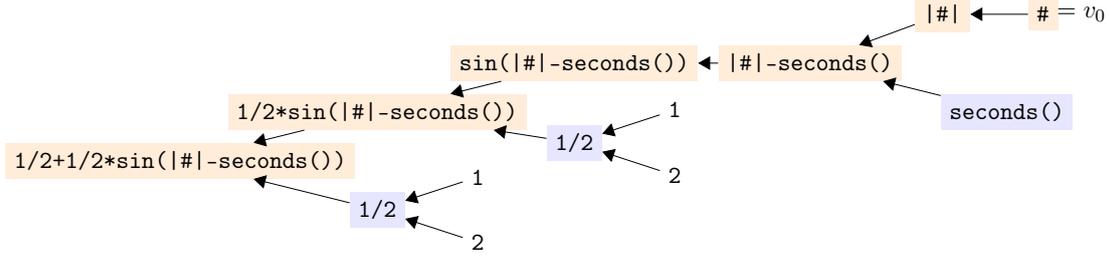

\subsection{Type detection}\label{type-detection}

A scripting language usually does not have a static type system. A variable can store any object; a function can return various values of different types. This noncommittal typing makes the fast-prototyping programming process easy, but the programs will be less efficient due to run-time type checking. However, for the transcompilation to the GPU variables, terms, arguments and return values of functions should take a static type, since the shader languages, trimmed for efficiency, only support static types.

We impose the following three requirements on the type system:
\begin{description}
  \item[No type annotations in the host language.] The programmer should not have to give type annotations. Omitting the type annotations is essential so that the GPU codes fit seamlessly into the scripting code of the host language. Possible up-casting of types should be done automatically.
  \item[No polymorphic types.] Because of the restrictions of the shading languages, which should run as efficiently as possible, only monotypes should be allowed. That means that a general type signature should not be inferred for any user-defined function. Instead, whenever a user-defined function is used, a single signature of monotypes is determined. Functions will be monomorphized.
  \item[Minimal typing.]  The types should be as weak as possible in order to make the program run fast. For example, if the transcompiler can prove that a variable remains to be an integer, then the integer data type on the GPU should be used instead of a floating point number or a complex number.

\end{description}

Thus, the transcompiler has to automatically assign static types to the variables and terms of a dynamically typed language. As is the case for Copperhead, this requires well-typed programs. A counter-example of a well-typed program is
\code{if(booleanexp, 12, [0])}, which returns \code{12} if \code{booleanexp} is \code{True} and \code{[0]} otherwise. This program will cause an error during the transcompilation since the types of the if and else branch cannot be unified to a single type.

We define the set of all types $T$ recursively:
\begin{itemize}
  \item $\bot$, $\top$, $\texttt{boolean}$, $\texttt{int}$, $\texttt{real}$, $\texttt{complex} \in T$
  \item $\forall \tau \in T \setminus \{\bot, \top\}, \forall n \in \mathbb{N}: list(n, \tau) \in T$
\end{itemize}
$\bot$ will correspond to the unset type and $\top$ to the error type, which will be used if no type could be determined. All other types can be modeled in a shader language. For instance, the type $list(4, list(4, \texttt{real}))$ corresponds to a $4\times 4$ matrix and can be modeled in \emph{GLSL} through the \texttt{mat4} type. Matrices of higher dimensions can be modeled with corresponding \texttt{structs} in \emph{GLSL}.
For dynamic geometry software, it is also suitable to add types such as $\texttt{point}$, $\texttt{line}$, and $\texttt{circle}$.

Furthermore, we equip $T$ with a reflexive, antisymmetric, and transitive \emph{subtype} relation $\sqsubseteq$ to make $(T, \sqsubseteq)$ a lattice, which is recursively generated as follows:
\begin{itemize}
  \item $\forall \tau \in T: \bot \sqsubseteq \tau \sqsubseteq \top$
  \item $\texttt{boolean} \sqsubseteq \texttt{int} \sqsubseteq \texttt{real} \sqsubseteq \texttt{complex}$
  \item $\forall \sigma,\tau \in T \setminus \{\bot, \top\}, \forall n \in \mathbb{N}: \sigma \sqsubseteq \tau \Rightarrow list(n, \sigma) \sqsubseteq list(n, \tau)$
\end{itemize}

A subtype relation $\sigma \sqsubseteq \tau$ means that every value that can be represented by type $\sigma$ can also be represented by a type $\tau$. Furthermore, each relation $\sigma \sqsubseteq \tau$ implies that there is an injective inclusion function from the set of values in $\sigma$ to the set of values in $\tau$. With $\sigma \sqcup \tau$ we denote the least-upper-bound of $\sigma$ and $\tau$. 

\begin{example} Typical values for $\sqcup: T \times T \to T$ are:%
\begin{eqnarray*}
\texttt{int} \sqcup \texttt{real} &=& \texttt{real} \\
\bot \sqcup \texttt{complex} &=& \texttt{complex} \\
list(5,\texttt{complex}) \sqcup list(5,\texttt{real}) &=& list(5,\texttt{complex}) \\
list(2,\texttt{real}) \sqcup list(3,\texttt{real}) &=& \top
\end{eqnarray*}
\end{example}

\begin{proposition}
$\sqcup: T \times T \to T$ is computable.
\end{proposition}
\begin{proof}
$\sigma \sqcup \tau$ for $\sigma, \tau \in T$ can be computed  recursively.
We can handle the cases $\{\sigma, \tau\} \cap \{\bot, \top\} \neq \emptyset$ via $\sigma \sqcup \bot = \sigma$, $\sigma \sqcup \top = \top$ (and the symmetric equations). If $\{\sigma, \tau\} \subset \{ \texttt{boolean}, \texttt{int}, \texttt{real}, \texttt{complex}\}$ then  $\sigma \sqcup \tau$ is the maximum of $\sigma$ and $\tau$ in the sequence $\texttt{boolean} \sqsubseteq \texttt{int} \sqsubseteq \texttt{real} \sqsubseteq \texttt{complex}$.
In all other cases, $\sigma$ or $\tau$ is a list. Wlog. let $\sigma = list(n, \alpha)$ for some type $\alpha\in T$. The term $list(n, \alpha) \sqcup \tau$ can be recursively evaluated as follows

\[
list(n, \alpha) \sqcup \tau = \begin{cases}
list(n, \alpha \sqcup \beta) & \text{ if $\tau = list(n, \beta)$}\\
\top & \text{ otherwise}
\end{cases}
\]
The recursion terminates because $\alpha$ and $\beta$ contain fewer list terms than $\sigma$ and $\tau$.

\end{proof}

   A variable that takes values of both type $\sigma$ and $\tau$ should be of type $\mu := \sigma \sqcup \tau$. Then $\sigma, \tau \sqsubseteq \mu$ and $\mu$ is minimal in this property.

Our primary aim is to determine the type of all variables and terms in $\E$. That is challenging because the program poses several, possibly circular, dependencies between the types of terms. First, we will model collections of types through product lattices. Then we will elaborate on how to model the dependencies between the different variables and terms as a mathematical condition on the product lattice. In the last step, we will show how a tailored fixed-point algorithm can be used to compute a minimal typing that fulfills the required conditions.

For any $n\in\mathbb{N}$, $T^n$ becomes a lattice by defining the product order
$(\sigma_1, \dots, \sigma_n) \sqsubseteq (\tau_1, \dots, \tau_n)$ iff $\sigma_i \sqsubseteq \tau_i$ for every $i \in \{1,\dots,n\}$. It can be easily seen that the height of the lattice $T$, and therefore the height of the lattice $T^n$, is finite. In particular they fulfill the \emph{ascending chain condition} (ACC), i.e. every increasing chain eventually becomes stationary.

\newcommand{\minsignature}{\textsc{minSign}}
A function $\textit{fun}$ of arity $n$ (i.e. $\textit{fun}$ takes $n$ arguments) can have multiple \emph{signatures}. In the transcompilation of $\textit{fun}$, it is suitable to choose the implementation of $\textit{fun}$ that has a signature as weak as possible, but as strong as necessary.
Therefore, for every $n$-adic function $\textit{fun}$, we equip the transcompiler with a function \[
\minsignature_{\textit{fun}}:  T^n \to T^n \times T
\] that, applied to the types of the provided arguments, returns a signature, which we consider as a tuple of argument types and a return type. If 
\[
\minsignature_{\textit{fun}}(\tau_1, \dots, \tau_n) = ( (\alpha_1, \dots, \alpha_n), \beta)
\]
we demand $(\tau_1, \dots, \tau_n) \sqsubseteq (\alpha_1, \dots, \alpha_n)$ and $(\alpha_1, \dots, \alpha_n) \to \beta$ is the \emph{minimal} signature of an available applicable GPU implementation of $\textit{fun}$ for arguments of type $(\tau_1, \dots, \tau_n)$, i.e. there is an implementation of $\textit{fun}$ that takes arguments of types $(\alpha_1, \dots, \alpha_n)$ and returns a value of type $\beta$. By \emph{minimal}, we mean that there is no other implementation with a signature $ (\alpha_1^*, \dots, \alpha_n^*) \to \beta^*$ such that $(\tau_1, \dots, \tau_n) \sqsubseteq (\alpha_1^*, \dots, \alpha_n^*) \sqsubset (\alpha_1, \dots, \alpha_n)$.
The map $\tau \mapsto \minsignature_{\textit{fun}}(\tau)_2$ is monotone.

\begin{example}\label{minsignatureexample}
Typical values for $\minsignature$ for the addition function $+$ and the square root $\sqrt{\cdot}$ can be 
\begin{eqnarray*}
\minsignature_{+}(\texttt{int}, \texttt{int}) &=& ((\texttt{int}, \texttt{int}), \texttt{int})\\
\minsignature_{+}(\texttt{complex}, \texttt{int}) &=& ((\texttt{complex}, \texttt{complex}), \texttt{complex})  \\
\minsignature_{\sqrt{\cdot}}(\texttt{int}) &=& ((\texttt{real}), \texttt{complex})
\end{eqnarray*}
For the first two values we assume that GPU implementations for $+$ for the data types \texttt{int}, \texttt{real}, and \texttt{complex} exist. This is mathematically motivated by the different domains for the functions $+ : \mathbb{Z}\times \mathbb{Z}\to\mathbb{Z}$, $+ : \mathbb{C}\times \mathbb{C}\to\mathbb{C}$ and $\sqrt{\cdot} : \mathbb{R}\to\mathbb{C}$. Whenever those functions are applied to values of a domain that can be embedded within another domain (i.e. $\mathbb{Z} \hookrightarrow \mathbb{R}$), then the generalized function on a wider domain can be used. This means that the evaluation of $+$ on two different arguments, the implementation for the lowest type that is a super-type of both its arguments should be chosen. If one of the two arguments is a real subtype of the other, it will automatically be upcasted before the evaluation.

The third example describes the (complex) square root. It is reasonable to provide an implementation that computes the complex square root of real numbers. If one wants to calculate the square root of an \texttt{int}, then it is suitable to choose the same function. It would also be possible to introduce additional datatypes for positive numbers, which can be used to ensure that the square root remains real.
\end{example}
Now, our aim is to compute a \emph{minimal} type assignment $\Gamma \in T^{\E}$ that assigns every relevant expression in $\E$ ($\E$ is  defined in Section \ref{detection-of-parts-for-parallelization} as the union of the dependent term/variables and the uniform variables) to a type of $T$ such that the assigned types are \emph{compatible} to the program in the following sense:
\begin{itemize}
  \item $\Gamma_s \sqsubseteq \Gamma_t$ if $t$ is a variable and there is an assignment $t:=s$ for a term $s$.
  \item $\minsignature_{\textit{fun}}(\Gamma_{a_1}, \dots, \Gamma_{a_n})_2 \sqsubseteq \Gamma_t$ if $t$ is a term that consists of the function $\textit{fun}$ applied to the arguments $a_1, \dots, a_n$, i.e. $t=fun(a_1, \dots, a_n)$.
  \item $\tau_0 \sqsubseteq \Gamma_{v_0}$ where $v_0$ is the running variable and $\tau_0$ its type (i.e. \texttt{int} for scenarios if \texttt{map} is applied to an integer list and $list(2, \texttt{real})$ if $v_0$ is a pixel coordinate).
  \item $\tau_u \sqsubseteq \Gamma_u$ if $u\in U$, i.e. $u$ is a uniform expression (this includes constant expressions), $u$ will be computed by the CPU, and is of type $\tau_u$.
\end{itemize}
Whenever a user-defined function with some arguments is called, then a ``virtual assignment'' for each of the argument variables is added, in order to make it possible to compute specific types of the arguments and establish a type assignment that also involves user-defined functions.

To compute a minimal compatible type assignment $\Gamma \in T^{\E}$, we will utilize the Kleene fixed-point theorem \citep{stoltenberg1994mathematical}: Let $\bot \in T^{\E}$ (i.e. $\bot=(\bot, \dots, \bot)$ -- no type is set yet), then we can iteratively apply a monotone function $F:T^{\E} \to T^{\E}$ on $\bot$. Since $T^{\E}$ fulfills the ACC, the ascending chain \[
\bot \sqsubseteq F(\bot) \sqsubseteq F(F(\bot)) \sqsubseteq \dots
\] eventually becomes stationary at a fixed point of $F$ and according to the Kleene fixed-point theorem, this stationary fixed point is the unique \emph{minimal} fixed point of $F$.

 In order to compute a minimal $\Gamma\in T^{\E}$ that fulfills the \emph{compatibility} requirements listed above, we choose the following monotone $F:T^{\E} \to T^{\E}$:
\[ 
F(\Gamma)_t =  \begin{cases}
 \Gamma_t \sqcup \bigsqcup_{i=1}^n \Gamma_{s_i} & \text{if $t$ is a variable and}\\
 & \text{there are the assignments $t:=s_1, \dots, t:=s_n$} \\
 \minsignature_{\textit{fun}}(\Gamma_{a_1}, \dots, \Gamma_{a_n})_2 & \text{if $t=fun(a_1, \dots, a_n)$} \\
\Gamma_t \sqcup \tau_0 & \text{if $t=v_0$ is the iteration variable} \\
\Gamma_t \sqcup \tau_u & \text{if $t=u$ is a uniform variable/term}
\end{cases}
\]
The function $F$ is component-wise monotone and hence also monotone in the product lattice. The component-wise monotonicity follows either by using the least upper bound of the input parameter or by the monotonicity of $\minsignature$.

Any fixed point of $F$ fulfills the \emph{compatibility} requirements listed above for $\Gamma$. 
Hence, we can compute a \emph{minimal} type assignment $\Gamma \in T^{\E}$ that fulfills the \emph{compatibility} requirements by iteratively applying $F$ to $\bot$ until it becomes stationary and we take this fixpoint of $F$ as $\Gamma$. This corresponds to a start with unset types and updating the types to the lowest suitable type whenever there is a place where the currently set types are not sufficiently expressible.

\begin{example}
Consider the example code 
\begin{minted}[fontsize=\small,mathescape]{CindyScript}
a = -2;
b = sqrt(a);
a = b + 1;
\end{minted}
\noindent It contains the terms $\E = \{ \mintinline{CindyScript}{a}, \mintinline{CindyScript}{-2}, \mintinline{CindyScript}{b}, \mintinline{CindyScript}{sqrt(a)}, \mintinline{CindyScript}{b+1}, \mintinline{CindyScript}{1}\}$. Using the presented fixed-point algorithm, we will determine the types of all the terms. We will start with a type assignment $\bot \in T^\E$ and iteratively apply $F$, which uses $\minsignature$ from Example \ref{minsignatureexample}:
  \begin{center}
    \begin{tabular}{r|cccccccc}
       & $\bot$ & $F(\bot)$ & $F^2(\bot)$ & $F^3(\bot)$ & $F^4(\bot)$ & $F^5(\bot)$ & $F^6(\bot)$ & $F^7(\bot)$ \\
       \hline
       \mintinline{CindyScript}{a} & $\bot$ & $\bot$ & \code{int} & \code{int} & \code{int} & \code{int} & \code{complex} & \code{complex} \\
       \mintinline{CindyScript}{b} & $\bot$ & $\bot$ & $\bot$ & $\bot$ & \code{complex} & \code{complex} & \code{complex} & \code{complex}\\
       \mintinline{CindyScript}{sqrt(a)} & $\bot$ & $\bot$ & $\bot$ & \code{complex} & \code{complex} & \code{complex} & \code{complex} & \code{complex}\\
       \mintinline{CindyScript}{b+1} & $\bot$ & $\bot$ & $\bot$ & $\bot$ & $\bot$ & \code{complex} & \code{complex} & \code{complex}\\
       \mintinline{CindyScript}{-2} & $\bot$ & \code{int} & \code{int} & \code{int} & \code{int} & \code{int} & \code{int} & \code{int} \\
       \mintinline{CindyScript}{1} & $\bot$ & \code{int} & \code{int} & \code{int} & \code{int} & \code{int} & \code{int} & \code{int} \\
    \end{tabular}
  \end{center}
Since $F^6(\bot) = F^7(\bot)$, this value gives a valid and minimal typing for the code.
\end{example}
\subsection{Transcompilation}\label{transcompilation}

Once a typing $\Gamma \in T^{\E}$ is computed, the transcompilation to GPU shader code is straightforward. If the program was not well typed, i.e. there is an expression $t\in \E$ with $\Gamma_t = \top$, then the computations will be evaluated on the CPU. Otherwise (and we hope in most cases), we apply the following scheme:

An expression $t=fun(a_1, \dots, a_n) \in D$ can be converted to shader code by carrying out the following three steps:
\begin{enumerate}[(1)]
  \item First, translate the arguments $a_1, \dots, a_n$ recursively to shader code.
  \item Then up-cast each of the translated arguments $a_i$, that has type $\Gamma_{a_i}$, to the type \[\left({\minsignature_{\textit{fun}}(\Gamma_{a_1}, \dots, \Gamma_{a_n})}_1\right)_i\] employing the subtype-embedding function. Note that by definition $(\Gamma_{a_1}, \dots, \Gamma_{a_n}) \sqsubseteq {\minsignature_{\textit{fun}}(\Gamma_{a_1}, \dots, \Gamma_{a_n})}_1$.
  \item Lastly, embed the implementation in the shader language of $\textit{fun}$ that corresponds to the signature $\minsignature_{\textit{fun}}(\Gamma_{a_1}, \dots, \Gamma_{a_n})$ into the header of the generated shader code (provided it has not already been done) and return the application of this function to the up-casted translated arguments as translated shader code for $fun(a_1, \dots, a_n)$.
\end{enumerate}
For each uniform variable/term, a new unique variable name will be generated and used in the program.

After this translation to shader code, the driver-dependent GPU compiler will be used to compile and link the generated shader code.

According to \cite{marrin2011webgl}, WebGL 1.0 allows only a fixed number of loop repetitions and fixed-length arrays. However, often the number of loop repetitions or the length of a list is based on input parameters. In this case of a changing number of repetitions, a re-compilation becomes necessary after the corresponding variables have changed their value. If the constant values are encoded as separate types, this re-compilation can be enforced without significant effort. A change of these constants would correspond to a type change, and, according to our scheme introduced in Section \ref{concept}, this would trigger a re-compilation.

\subsection{Lazy storage of data}\label{lazysec}
A suitable way to store the outcomes of shader programs is to write to textures. Reading this data is possible by texture lookups. According to \cite{gregg2011data}, the memory transfer between CPU and GPU has severe effects on the running time of an application and thus has to be minimized. Therefore, we store GPU-generated data exclusively on the GPU as long as necessary, and the data is only visible for shader programs that run on the GPU. Only if there are read accesses by the part of the code that runs on the CPU, data will be transferred.

Within a single shader program, the preference is often to write data to a target (texture) that is also used for reading. That is problematic for many GPU-APIs. For instance, the WebGL API specifies that the occurrence of operations that both write to and read from the same texture, creating a feedback loop, will generate an error, see \cite[6.26]{marrin2011webgl}. For these feedback loops, a ping-pong approach can be used: If both read and write accesses on a texture object are detected, the texture will be stored twice: one texture for reading and another target texture for writing. After the execution of a shading program that writes to the corresponding texture, the two textures will be swapped. For following reading accesses, data will be read from the generated texture.

\section{Example Implementation: \emph{CindyGL}}\label{example-implementation-cindygl}

In \citep{montag2016cindygl}, we demonstrated an implementation of our proposed concepts.
We developed a plug-in called \emph{CindyGL} for \emph{CindyJS}. \emph{CindyJS} \citep{von2016cindyjs} aims to be a open source web-compatible porting of the dynamic geometry software \emph{Cinderella} \citep{richter2000user} \footnote{The entire project project is open source and available at \url{https://github.com/CindyJS/CindyJS}}.
\emph{CindyGL} can translate the \emph{Cinderella} inherent untyped scripting language \emph{CindyScript} to \emph{GLSL} and it enables a smooth integration of dynamic geometry, CPU, and GPU programming.

As \emph{CindyJS} runs in a web environment, \emph{CindyGL} can utilize WebGL to execute the compiled \emph{GLSL}-Code on the GPU.
This capability leads to easy portability because WebGL-capable browsers are widespread.  By August 2018, almost all\footnote{According to
\url{https://webglstats.com/}, $98\%$ in August 2018} of the browsers used on desktops, smartphones and tablets supported WebGL 1. Visualizations that use \emph{CindyGL} can be quickly distributed because no additional software usually has to be installed to access WebGL-based contents. Also, the visualizations will be comfortably accessible on new devices that provide a WebGL-capable browser. If no acceleration is available through the GPU, as a fallback solution, the smooth language integration can be utilized to execute all the code on the CPU.
\newcommand{\C}{\mathcal{C}}

The comparable simplicity to program GPU accelerated applications through \emph{CindyGL} has led to several implementations. Examples of visualizations that have been generated through \emph{CindyGL} can be found online\footnote{For example, in our web-gallery \url{https://cindyjs.org/gallery/cindygl/}} and tested on almost every browser. In the following sections, we will present a short introduction to \emph{CindyGL}\footnote{A more elaborate tutorial is available online at \url{https://cindyjs.org/docs/cindygltutorial/}} and briefly demonstrate some \emph{CindyGL}-applications.

\subsection{Usage of \emph{CindyGL}}

\emph{CindyGL} implements a command called \mintinline{CindyScript}{colorplot}. The \mintinline{CindyScript}{colorplot} command assigns a color to each pixel of the screen or a given area according to a given function.  The given function is usually dependent on the pixel coordinate \mintinline{CindyScript}{#}. \mintinline{CindyScript}{#} is a 2-component vector. Alternatively, if another free variable is detected, it can also become the running variable. If both the variables \mintinline{CindyScript}{x} and \mintinline{CindyScript}{y} are free, then \mintinline{CindyScript}{colorplot} will interpret \mintinline{CindyScript}{x} and \mintinline{CindyScript}{y} as the coordinates of a pixel. If the free variable \mintinline{CindyScript}{z} is used, then the coordinate will be interpreted as a complex number with \mintinline{CindyScript}{z = x+i*y}.

If the function within the \mintinline{CindyScript}{colorplot} statement attains real numbers as values, then its color will be interpreted as a grayscale value where $0$ is black, and $1$ is white. If the result is a three-component vector, then the vector will be interpreted as an RGB value. If there is also a fourth component, the value of this component will be interpreted as an alpha value (transparency).

\begin{example} \label{waveexample}
The following \emph{CindyScript} code is an example to render a circular sinusoidal wave through a \mintinline{CindyScript}{colorplot}:
\begin{minted}[fontsize=\small,mathescape]{CindyScript}
colorplot( // assigns to each pixel with coordinate # a color
 1/2+1/2*sin(|#|-seconds()) // an animated centered sinusidial wave
);         // seconds() is the current time in seconds.
\end{minted}
\noindent Once \emph{CindyJS} executes the code, a circular wave as in Fig. \ref{fig:cindyglapps1}a becomes visible. By executing the code many times within a second, an animation is created. Since the computation is accelerated on the GPU, real-time rendering is possible on almost every modern device.
\end{example}
\begin{figure}
  \begin{center}
    \begin{adjustbox}{width=\textwidth}
    \includegraphics[height=3cm]{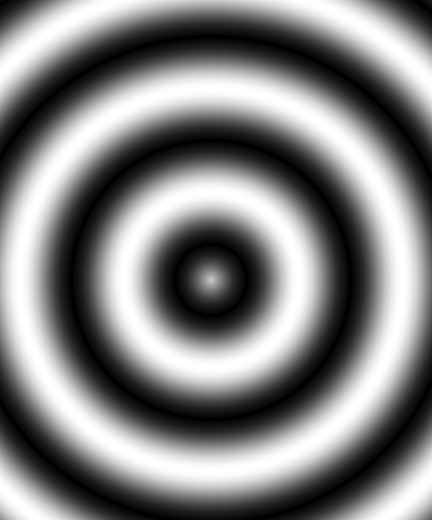} \hfill
    \includegraphics[height=3cm]{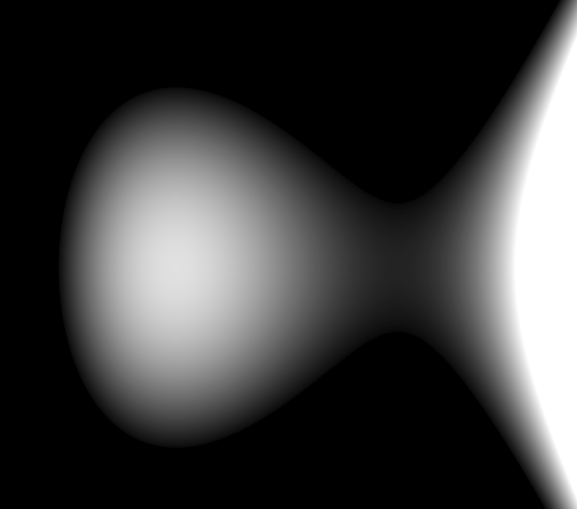} \hfill
    \includegraphics[height=3cm]{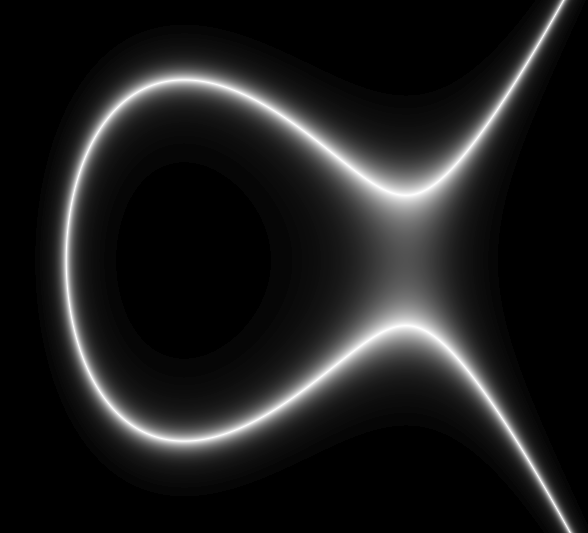} \hfill
    \includegraphics[height=3cm]{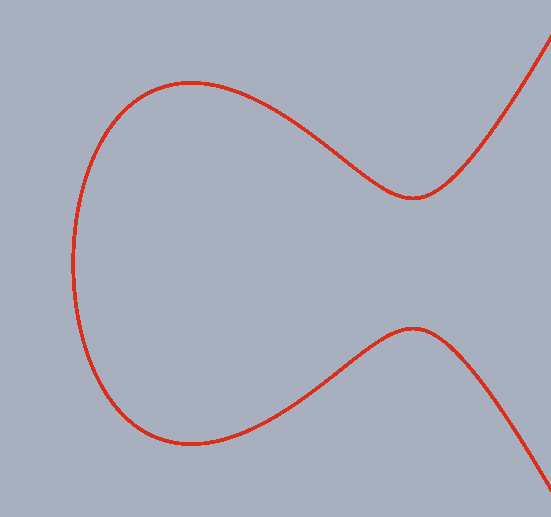}
    \end{adjustbox}
    \caption{(a) Circular sinusoidal wave generated by the \mintinline{CindyScript}{colorplot} command. (b), (c), (d): Visualizations of the elliptic curve $y^{2} = x^{3}-x+\frac{1}{2}$ as zero set of $f(x,y) = x^{3}-x+\frac{1}{2} - y^{2}$: (b) \mintinline{CindyScript}{colorplot(f(P))} making the values of $f(P) \in [0,1]$ distinguishable as grayscale value (values outside the range are either black or white), (c) \mintinline{CindyScript}{colorplot(exp(-10*|f(P)|))} highlighting the zero set and (d) by an approach that detects points close to the curve based on the intermediate value theorem.}
    \label{fig:cindyglapps1}
  \end{center}
\end{figure}

\emph{CindyGL}, as an implementation of the concepts described in Section \ref{concept}, compiles the suitable parts of the code which is passed to \mintinline{CindyScript}{colorplot} to \emph{GLSL} and then to WebGL-binary the first time the command is called on a given expression. For consecutive calls, the compiled program is re-used as long as the occurring types have not changed. If the types have changed in the meantime, a recompilation is forced.

In the following sections, we will present a selection of various use case scenarios of \emph{CindyGL}. %
\subsection{Implicit curves and sets of locus within dynamic geometry software}

\newcommand{\R}{\mathbb{R}}
We assume that we are in the Euclidean plane and $f:\R^2\to \R$ is a smooth function. Instead of a plot of $f$, often the variety $V(f) := f^{-1}(\{0\})$ is of special interest. In this section, we assume that the gradient of $f$ does not vanish on $V(f)$, which makes $V(f)$ (locally) a curve by the implicit function theorem. The \mintinline{CindyScript}{colorplot} command gives a very simple pipeline to visualize these curves. It is still very efficient compared to conventional approaches (for no known parametrizations) because the GPU performs the computations in parallel.

An example of implicit curves are elliptic curves. They can be defined by the equation
\[
y^{2}=x^{3}+a x+b
\]
where $a$ and $b$ are real numbers such that $\Delta =-16(4a^{3}+27b^{2})\neq 0$.

We define an implicit $f$ as follows
\begin{minted}[fontsize=\small,mathescape]{CindyScript}
f(P) := (
  x = P.x; y = P.y;
  x^3 + a*x + b - y^2 // last line: return value
);
\end{minted}
\noindent It is clear that $f(P)=0$ iff $P$ is in the elliptic curve.
The evaluation of $\mintinline{CindyScript}{colorplot(f(P))}$ renders an image as in Fig. \ref{fig:cindyglapps1}b ($P$ is a free variable and therefore will be detected as a varying pixel coordinate). The elliptic curve is located where black just becomes gray.
How can we visualize the points $P$ in the plane such that $f(P)$ becomes zero more apparently? In general, running a test to determine whether $f(P)$ becomes zero for a finite set of pixel coordinates $P$ obtained by rasterization is problematic because of numeric issues and the fact that it is improbable that $P$ will hit a zero of $f$, even if the corresponding (rectangular) pixel contains zeros of $f$ because if $f$ is a non-constant algebraic function, the zero set has Lebesgue measure 0.
 \cite{liste2014color}, who utilizes the ``sweeping-line'' GeoGebra on the CPU, suggests a rendering approach, which can be modeled in CindyGL via the command \mintinline{CindyScript}{colorplot(exp(-10*|f(P)|))} (see Fig. \ref{fig:cindyglapps1}c). All of Liste's applications can be easily transferred to the modern GPU-based \mintinline{CindyScript}{colorplot} via \emph{CindyGL} and work in real-time on this architecture.
 
A straightforward approach to determine a subset of pixels that contain a part of the implicit curve is the evaluation of the values of $f$ at each corner of every pixel. If these values can be separated by zero, then according to the intermediate value theorem, the curve must ``enter'' the pixel, and the pixel can be marked as one that contains the curve. This approach can be implemented in CindyGL via
\begin{minted}[fontsize=\small,mathescape]{CindyScript}
tinysquare = [[-1,-1],[-1,1],[1,-1],[1,1]]/100; // the corners of a pixel
colorplot(
  // evaluate f at the corners of pixel with center P
  values = apply(tinysquare, delta, f(P+delta));
  if(min(values) <= 0 & 0 <= max(values),
     [1,0,0,1], // the signs swap => plot red with full alpha
     [0,0,0,0]  // all signs are same => plot nothing (transparent)
  );
);
\end{minted}
\noindent If the gradient of $f$ does not vanish at $V(f)$ (for the elliptic curves this is guaranteed by $\Delta \neq 0$) and if the curve has a sufficiently low curvature, then all pixels containing a substantial part of the curve are detected by this method, and a meaningful image can be acquired. Replacing \mintinline{CindyScript}{tinysquare} with a bigger $n$-gon,  provides a straightforward approach to render the curve less precise, but bolder (see Fig. \ref{fig:cindyglapps1}d).

\begin{figure}
  \begin{center}
    \begin{adjustbox}{width=\textwidth}
    \includegraphics[height=5cm]{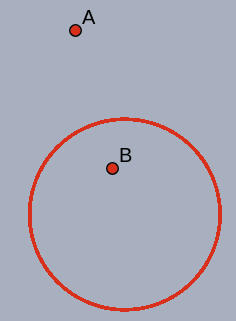} \hfill
    \includegraphics[height=5cm]{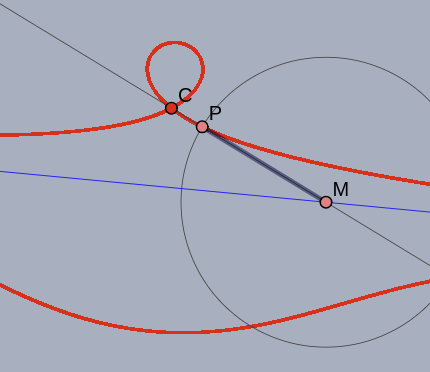} \hfill
    \includegraphics[height=5cm]{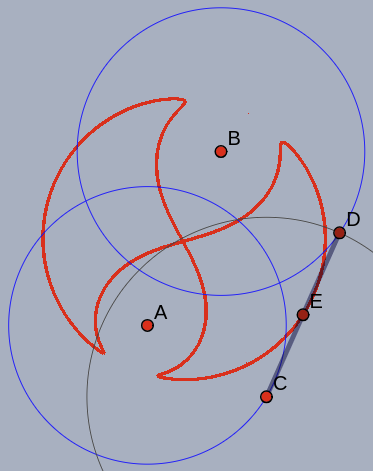} \hfill
    \end{adjustbox}
    \caption{(red) Loci computed on the GPU via the \mintinline{CindyScript}{colorplot} command: (a) All points $P$ such that $\frac{|P-A|}{|P-B|}=2$ (b) The conchoid of Nicomedes, i.e. the set of all points $P$ that are obtained by moving $M$ in the construction. The second branch of the conchoid can be reached if $M$ is moved through infinity
    (c) Watt's curve, which describes the set of all points $E$ such that $E$ is the midpoint of $C$ and $D$, which are two points on distinct circles and have a constant distance $|C-D|$.
    }
    \label{fig:cindyglloci}
  \end{center}
\end{figure}

These implicit curves are of particular interest in dynamic geometry software if they describe the locus of points by geometric terms.
As \cite{botana2014automatic} point out, the computation of loci can be a tool for \emph{computer-aided discovery} of mathematical properties. As a very simple example, given two distinct points $A$ and $B$ in the Euclidean plane and a ratio $r\in \mathbb{R}_{>0}$, a user of DGS could be interested in the set of all points $P$ such that the ratio $\frac{|P-A|}{|P-B|}$ takes the value $r$. The set of those points can be visualized by the approach above by exchanging $f$ with
\begin{minted}[fontsize=\small,mathescape]{CindyScript}
f(P) := |P-A|/|P-B|-r;
\end{minted}
\noindent The rendered curve, which is the zero set of $f$, are all the points fulfilling $\frac{|P-A|}{|P-B|}=r$ (Fig. \ref{fig:cindyglloci}a shows an image for $r=2$). When rendering this curve, it instantly becomes visible that those lines of an equal ratio of distances are circles, or if $r=1$, the perpendicular bisector (line) of $A$ and $B$. Certainly, this observation of an image does not give a formal proof. However, it calls to find a proof of such properties.

A more complex locus is the conchoid of Nicomedes. Let $a$  be a line, and $C$ be a point. Furthermore, let $C_0$ be a circle with its center $M$ on $a$ and a fixed radius. The conchoid is the set of all points $P$, that can be obtained as the intersection of $C_0$ and the line connecting $M$ and $C$ if $M$ is moved along $a$. This condition can also be written as an implicit formula for the point $P$ by building a backward construction of $M$ based on a point $P$. The point $M$ can be constructed as the intersection of $a$ and the line connecting $C$ and $P$. $P$ is contained in one of the two branches of the conchoid if $P$ lies on a circle with $M$ as the midpoint and the radius $C_0$. The geometric construction and the check give rise to the following definition of $f$, that in a sense reverses the construction sequence:
\begin{minted}[fontsize=\small,mathescape]{CindyScript}
f(P) := (
  l = join(C, P); // the line connecting C and P
  M = meet(l, a); // the intersection of l and a
  |P.xy-M.xy|-C0.radius // return value; $=0$ $\Leftrightarrow$ P is in the locus
);
\end{minted}
\noindent The zero set of $f$, which is the locus set of the conchoid, can be rendered by the approach above (omitting potential singular points). The results can be seen in Fig. \ref{fig:cindyglloci}b. If the geometric construction, which is built within the DGS, is modified, then the GPU computes the new locus in real-time. Note that also types for geometric primitives, such as points and lines, and basic geometric operations have been implemented as a data-type in CindyGL. This enables the transcompilation of such functions $f$ to GPU code and thus enables a stronger interplay of dynamic geometry software and GPU computations.

We were also able to compute the locus of Watt's curve on the GPU (see Fig. \ref{fig:cindyglloci}c) utilizing a similar geometric ``backward construction''. 
The construction involves the intersection of circles and the reflections around a point, which again has been transcompiled to the GPU by our approach and the images of the locus are rendered in real-time. The generalization and mechanical application of this technique of using a backward construction to yield an implicit formula could lead to further research.

\subsection{Raycasting algebraic surfaces}

\begin{figure}
  \begin{center}
    \begin{adjustbox}{width=\textwidth}
    \includegraphics[height=4cm]{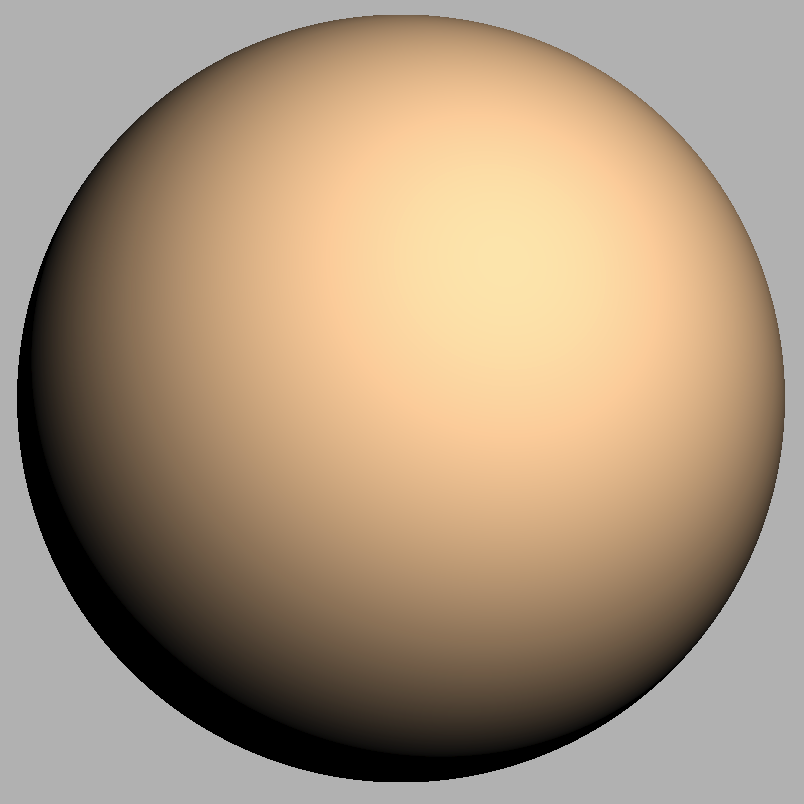} \hfill
    \includegraphics[height=4cm]{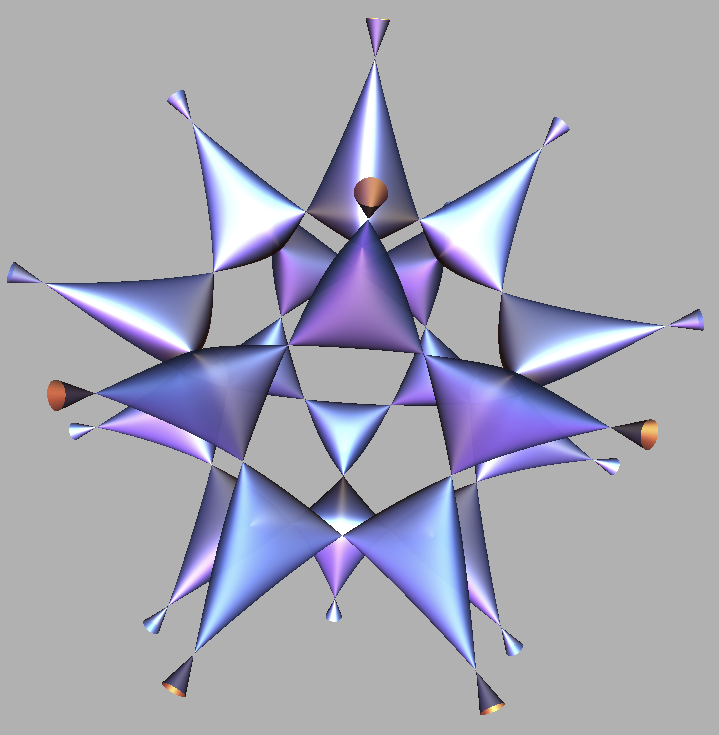} \hfill
    \includegraphics[height=4cm]{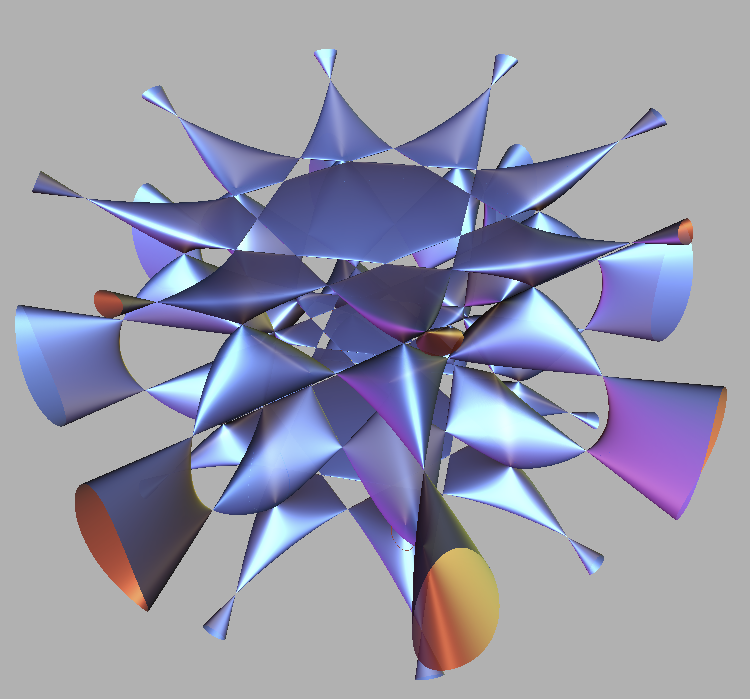} \hfill
    \end{adjustbox}
    \caption{Algebraic surfaces raycasted via the \mintinline{CindyScript}{colorplot} command: (a) the sphere (b) Barth Sextic and (c) Endraß Octic}
    \label{fig:cindyglraycasted}
  \end{center}
\end{figure}

Motivated by the software \emph{SURFER} \citep{stussak2009realsurf} and the fact 
that it is likely that this Java-based platform will no longer be supported (i.e. maintainable) in a few years, a further design goal of CindyGL was to provide user-friendly machinery that can be used to implement algorithms for rendering algebraic surfaces.
In the following, let $F:\R^3\to\R$ be algebraic and irreducible. The aim is to visualize the surface $V(F) = F^{-1}(\{0\})$ by raycasting.

Mathematically, the process of raycasting can be described as follows: For each pixel $p$ of a rasterization screen, there is a ray (an affine linear function) $r_p:\R\to\R^3$  that passes through the pixel on an imaged projection plane and possibly ``hits'' the surface $V(F)$ at some coordinate $s_p \in R^3$ (possible within a predefined clipping area). The pixel $p$ can be colored according to $s_p$. In order to obtain good visualizations, the normalized gradient $\frac{\Delta F(s_p)}{\|\Delta F(s_p)\|}$ can be taken into consideration; the shading can be done based on the scalar product of a light direction and the normal of the surface. An important observation is that the function $f_p = F \circ r_p: \R \to \R$ is a polynomial and the degree of $f_p$ is bounded by the degree of $F$. A root $t$ of $f_p$ corresponds to the root $r_p(t)$ of $F$. Hence, the problem of raycasting algebraic surfaces can essentially be reduced to the problem of locating the real roots of a univariate polynomial for each pixel. This scheme is suitable for massively parallel computations on the GPU.

\begin{example}
A simple example is the function $F(x,y,z)=x^2+y^2+z^2-1$ with $V(F)=S^2$. For simplicity, we assume that a pixel $p$ with coordinates $(x,y)$ is associated with the ray $r_p(t) = (x, y, t)$ and we yield the quadratic polynomial $f_p(t) = F \circ r_p(t) = t^2+(x^2+y^2-1)$, having the real roots $\pm \sqrt{x^2+y^2-1}$ if $x^2+y^2-1\geq 0$. Let $s_p$ be the first intersection of the ray $r_p$. It exists if $x^2+y^2-1\geq 0$ and has the value $s_p = r_p(-\sqrt{x^2+y^2-1}) = (x,y, -\sqrt{x^2+y^2-1})$. Since $s_p \in S^2$, we obtain the normal $\frac{\Delta F(h_p)}{|\Delta F(s_p)|} = s_p$ for free. This yields the following \emph{CindyScript} code that renders the sphere:
\begin{minted}[fontsize=\small,mathescape]{CindyScript}
lightdir = [.3,.4,-1]; lightcolor = [1,.8,.6]; background = [.7,.7,.7];

colorplot(
  if(1-x^2-y^2>=0,
    s = (x,y,-|sqrt(1-x^2-y^2)|); // intersection with the sphere
    (s*lightdir) * lightcolor, // shading based on normal s
    background
  )
);
\end{minted}
\noindent CindyGL can transcompile the code within the \mintinline{CindyScript}{colorplot} command to GPU shader code, and an image (see Fig. \ref{fig:cindyglraycasted}a) of the sphere is rendered on the GPU.
\end{example}

For general algebraic surfaces of higher degree, the roots of the functions $h_p$ have to be found using a numerical method. We followed the approach of \cite{stussak2009realsurf} and got an easy implementation for raycasting ``tame'' algebraic surfaces via \emph{CindyGL}. For each ray $r_p$, we evaluate $\deg F + 1$ values of $F$ along the ray $r_p$ and obtain the polynomial $f_p$ by interpolation (this avoids symbolic computations). Descartes's rule of signs is used to isolate the real roots of the univariate polynomials $f_p$. We represent the polynomials $f_p$ in the Bernstein basis to enhance the numeric stability. Furthermore, the Bernstein basis further simplifies the implementation because the number of sign variations needed to apply Descartes's rule corresponds to the number of sign variations of the coefficients in Bernstein basis with respect to the investigated interval \citep{sagraloff2016computing}. Once the real roots are isolated, the bisection method is used to approximate the points where $r_p$ intersects the surface. Our raycaster is available online\footnote{\url{https://cindyjs.org/gallery/cindygl/Raytracer/}} and renders visualizations of several algebraic surfaces in real-time (see Figs. \ref{fig:cindyglraycasted}b and \ref{fig:cindyglraycasted}c for examples).

\subsection{Editing (spherical) images in real-time and the spherical Droste effect}

\emph{CindyJS} can also read images, the live feed from a webcam or a spherical camera. \emph{CindyGL} can apply mathematical transformations on these images in real-time.

\cite{mercatgagern2010} introduce several mathematical image transformations (``filters'') that are suitable for the GPU and present applications the fields of education and the arts. The filters include the generation of (hyperbolic) tilings by wallpaper groups and the deformation of images using conformal maps. von Gagern and Mercat note that using live footage, such as from a webcam, could be used to demonstrate mathematical concepts to the masses in an appealing and fascinating way and therefore could raise people’s interest in the underlying mathematical concepts (``edutainment'').
The filters presented in their work can be implemented in \emph{CindyGL} with relatively little effort. Mercat has already transferred his \emph{conformal webcam} to \emph{CindyGL}\footnote{See \url{http://bit.ly/webcamconf}}.
\begin{figure}
\centering
    \includegraphics[height=5.5cm]{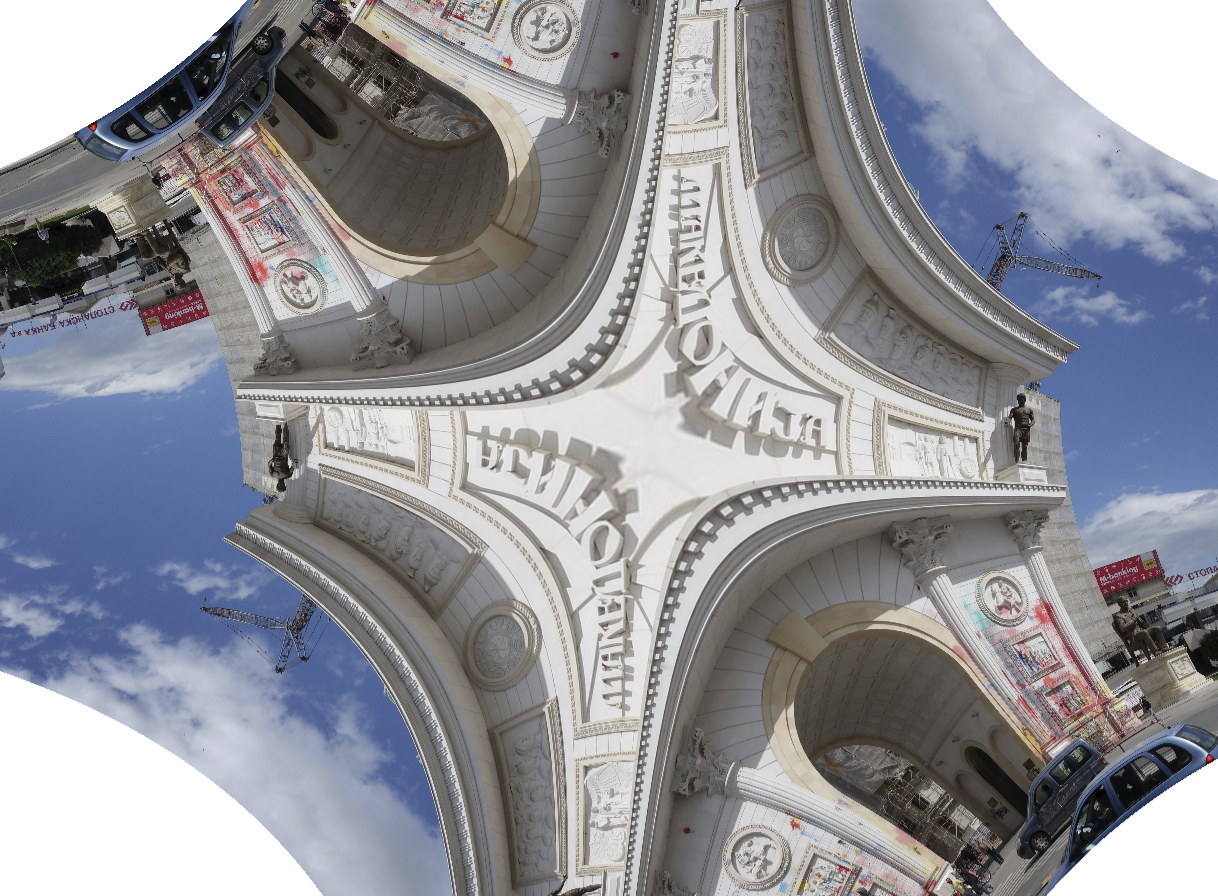} \hfill
    \includegraphics[height=5.5cm]{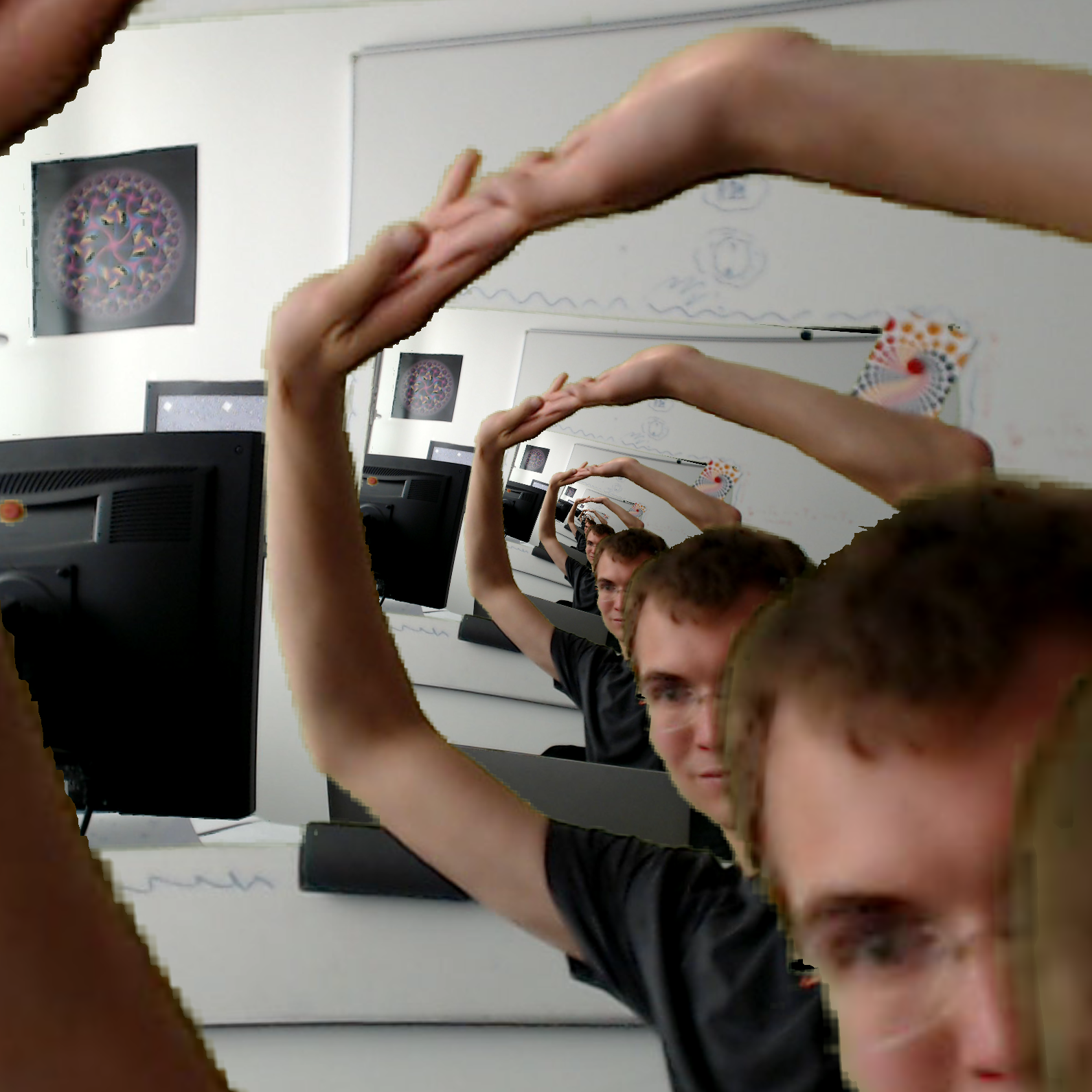}
    \caption{Deformation of images through CindyGL: (a) An image of the Porta Macedonia that has been deformed by squaring of the complex coordinates before looking up the colors in the image (b), The twisted Droste effect applied on (spherical) live footage showing one of the authors in front of green screen}
    \label{fig:deform}
\end{figure}
\begin{example} The one-line program \mintinline{CindyScript}{colorplot(readimage(cameravideo(), z^2));} can be used to set up a \emph{conformal webcam} for the function $f:\mathbb{C} \to \mathbb{C}$, $f(z)=z^2$. Each pixel with the complex coordinate $z$ on the screen gets assigned the (interpolated) pixel color from the position $z^2$ of the webcam. Images as in Fig. \ref{fig:deform}a are generated in real-time. %
\end{example}

\cite{schleimer2016squares} proposed to use Möbius transformations as a tool to edit spherical images. Holomorphic functions can be applied to spherical images if the image either is interpreted as the Riemann sphere or if the image - which is an approach for computation - is mapped to the complex plane via the stereographic projection\footnote{An implementation and visualization of this concept based on \emph{CindyGL} is available online at \url{https://montaga.github.io/stereographic/}}.
Based on the work of \cite{de2003mathematical}, they demonstrated how twisted Droste images could be generated from spherical video footage with the complex exponential, logarithm, and Möbius transformations. Using the ability of \emph{CindyJS} to capture the spherical webcam as an input stream, and applying several mathematical operations on spherical footage in real-time in \emph{CindyGL}, we have implemented an application\footnote{It is also available online at \url{https://montaga.github.io/droste/}} that renders the spherical twisted Droste effect on a green screen (see Fig. \ref{fig:deform}b). The stand-up comedian and mathematics communicator Matt Parker used this spherical Droste application and other \emph{CindyGL} applications for multiple live performances during a tour\footnote{\emph{CindyGL} was used in the ``You Can’t Polish A Nerd'' tour during the ``Festival of the Spoken Nerd'': \url{http://festivalofthespokennerd.com/show/ycpan/}}.

\subsection{Feedback loops and GPGPU applications in \emph{CindyGL}}

A strength of \mintinline{CindyScript}{colorplot} is the ability to plot to a texture besides drawing on the screen.
This rendered texture can be read in consecutive calls of \mintinline{CindyScript}{colorplot} with the \mintinline{CindyScript}{imagergb} command. The possibility to read and write texture data enables the creation of feedback loops on the GPU and opens the door for many GPGPU computations.

If a user reads and writes to the same texture with some deformations in between, video feedback loop can be simulated by pointing a camera at a screen that displays the image that the camera records. This technique presents an interesting approach to rendering fractals:
\begin{example}
Fractals that utilize an escape time algorithm can be rendered through a feedback loop system via \mintinline{CindyScript}{colorplot}. For instance, the filled Julia set for a function $z \mapsto z^2+c$ can be approximated by running the following \emph{CindyScript} source code several times:
\begin{minted}[fontsize=\small,mathescape]{CindyScript}
colorplot("julia", // plot to texture "julia":
 if(|z|<2,         // if $\lvert z\rvert < 2$, take the color from texture
  imagergb("julia", z^2+c) // "julia" at position $z^2+c$ 
   + (0.01, 0.02, 0.03),   // and make it slightly brighter
  (0, 0, 0)        // otherwise: display black.
 )
);
\end{minted}
\noindent The color of a pixel becomes brighter if it takes more iterations to leave the escape radius $2$ from the coordinate of the pixel. After about 50 iterations the fractal in Fig.\ref{fig:feedback}a can be seen.
\end{example}

\begin{figure}
  \begin{center}
    \begin{adjustbox}{width=\textwidth}
    \includegraphics[height=4.3cm]{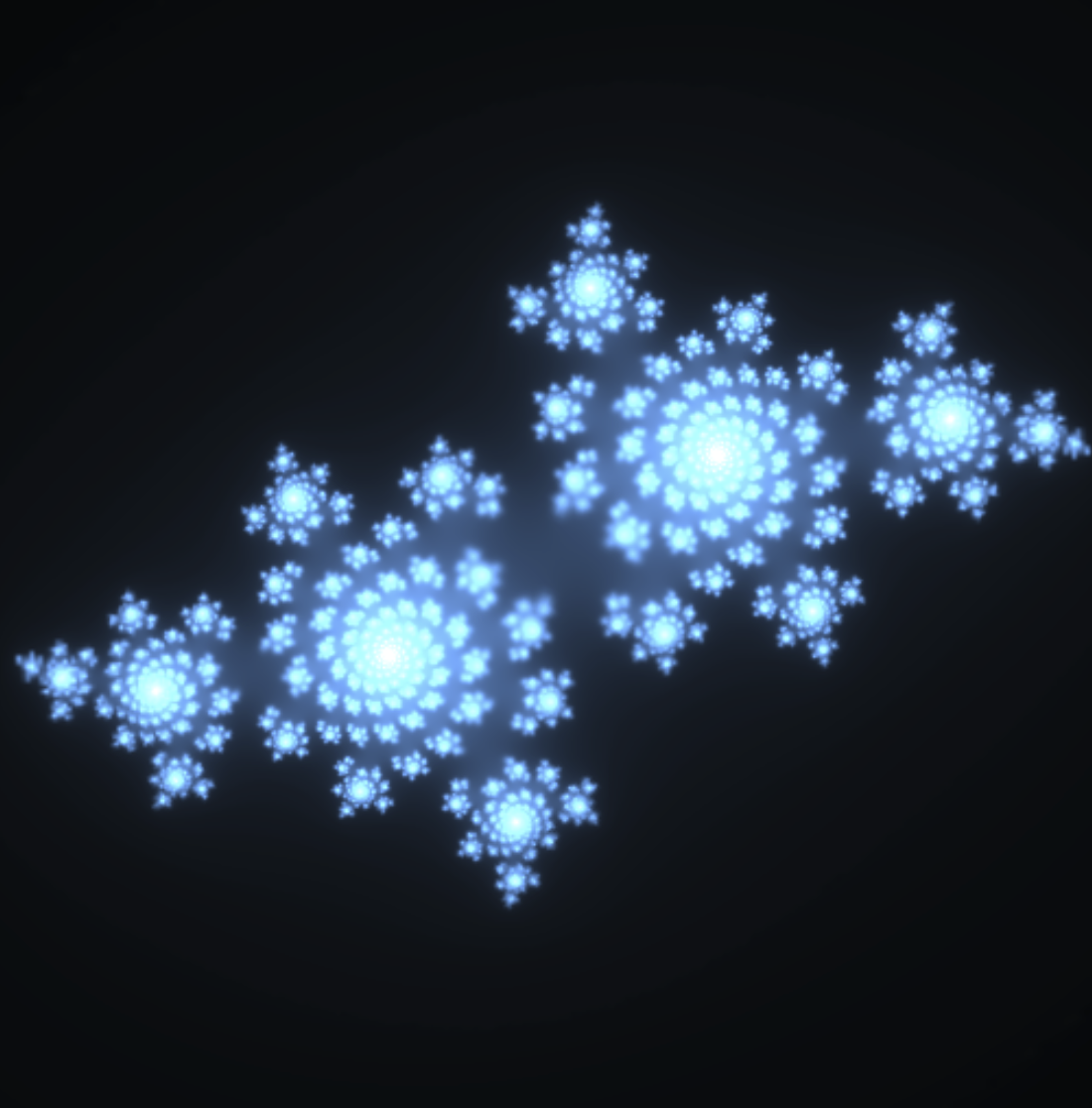} \hfill
    \includegraphics[height=4.3cm]{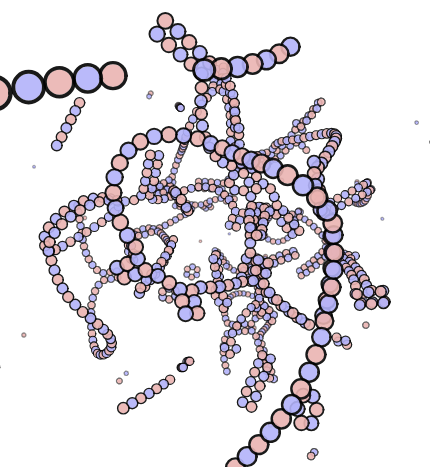} \hfill
    \includegraphics[height=4.3cm]{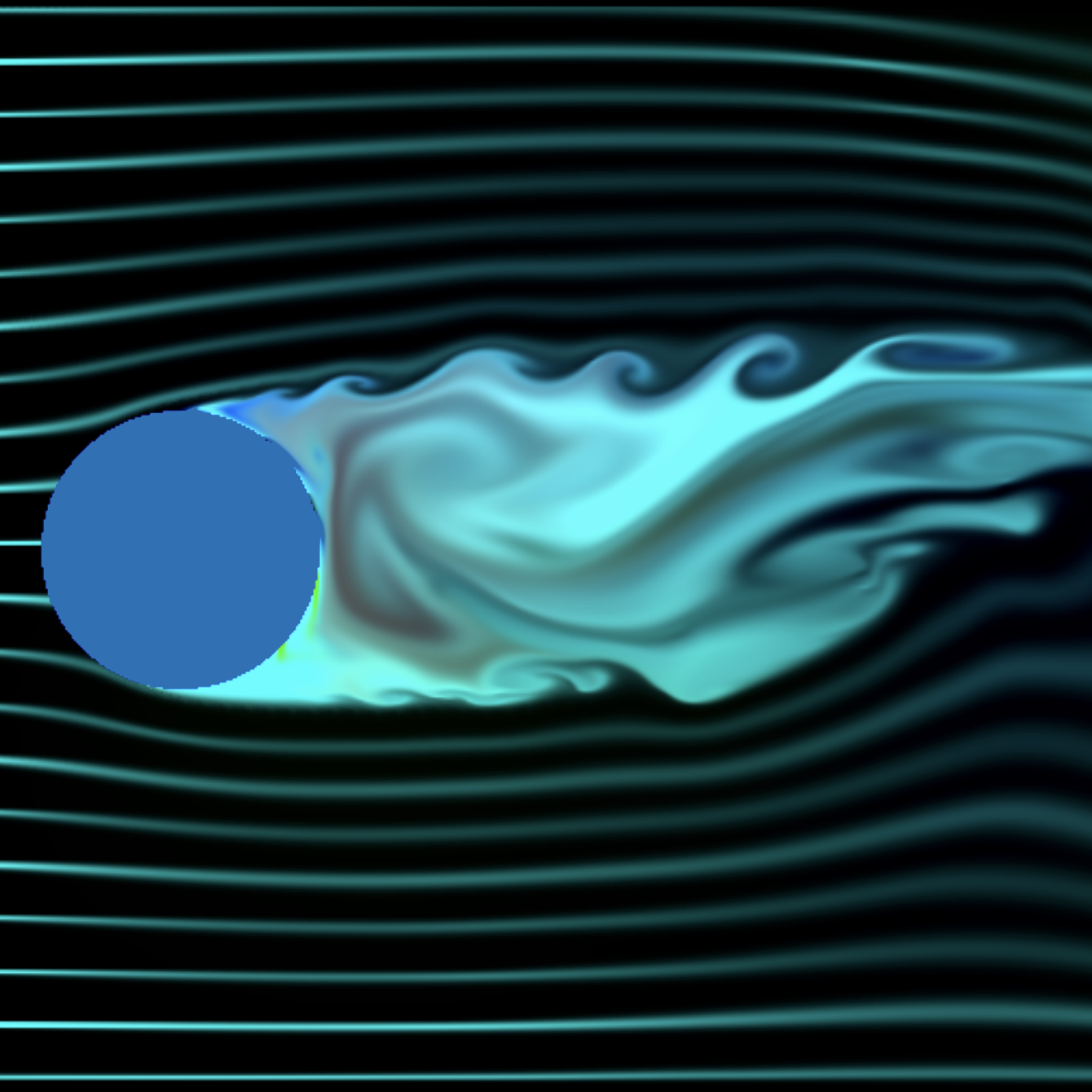}
    \end{adjustbox}
    \caption{(a) The Julia set generated through a feedback loop approach (b) Numeric simulation of the n-body problem with $500$ particles with positive unit charge and $500$ with negative unit charge. All the $n^2=10^6$ interactions (Lennard-Jones potential) are computed and provide a real-time simulation. (c) Numeric simulation of the Navier-Stokes equations}
    \label{fig:feedback}
  \end{center}
\end{figure}

The use of \emph{CindyGL} is not limited to graphical visualizations. Computationally demanding numerical schemes are often good candidates for the GPU as well. Data can be stored on textures, and parallelized computations on the data can be triggered by executing a \mintinline{CindyScript}{colorplot} that reads from these textures via \mintinline{CindyScript}{imagergb}. As an example for GPGPU programming, we used \emph{CindyGL} to simulate the interactions of $n$ bodies (see Fig. \ref{fig:feedback}b). The maximum number of simulated particles that allows a real-time visualization increased drastically compared with a corresponding CPU implementation. We also simulated the behavior of a fluid by approximating the solution of the Navier-Stokes equations on a GPU (see Fig. \ref{fig:feedback}c). The total number of lines of code required to implement this simulation applet with \emph{CindyGL} dropped to about one fourth as compared to an equivalent plain WebGL implementation.

\subsection{Educational value}{

\cite{kaneko2017} argues that dynamic geometry software on personal computers is often not used in classrooms because of technical obstacles. However, providing the teaching content on modern touch devices such as iPads and tablets can overcome this obstacle. Aside from PCs, both \emph{CindyJS} and \emph{CindyGL} are suitable for these modern devices because they only depend on a portable plugin-less web technology. Kaneko showed that using this technology has a positive influence on learning.

So, for teachers, the use of \emph{CindyGL} provides a suitable technique to provide GPU-accelerated content. Using \emph{CindyGL}, the instructors can distribute the interactive content by sharing a single HTML file.

For the students, the task to generate visualization with this framework can be very fruitful as well. Different learning content from mathematics, computer science, physics and other fields can be connected in such tasks. Because of the seamless integration of GPU code in a familiar scripting environment of a DGS, the creation process does not require the user to know or learn another programming language. \emph{CindyGL} eases the programming of shaders. Non-experts can use it to visualize mathematical concepts without having to learn about shader programming.

\begin{figure}
  \begin{center}
    \begin{adjustbox}{width=\textwidth}
    \includegraphics[height=3.5cm]{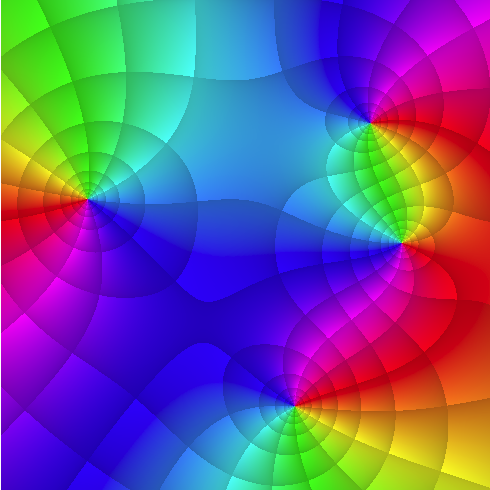} \hfill
    \includegraphics[height=3.5cm]{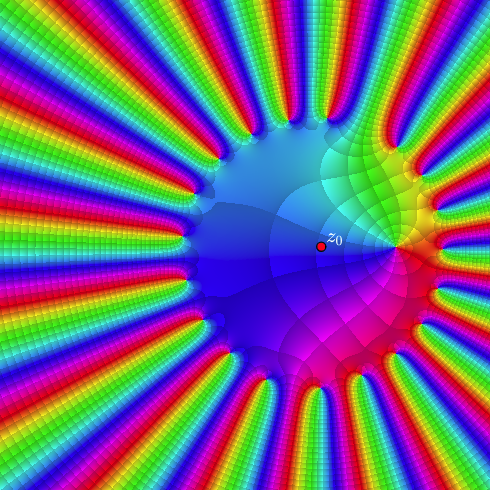} \hfill
    \includegraphics[height=3.5cm]{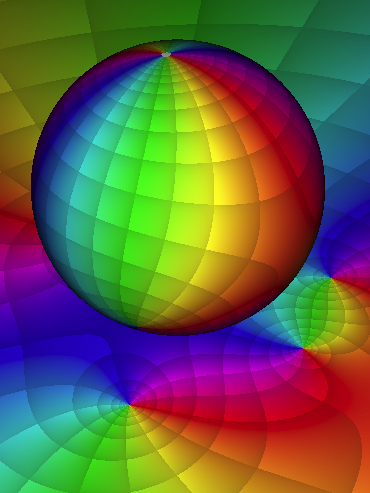} \hfill
    \includegraphics[height=3.5cm]{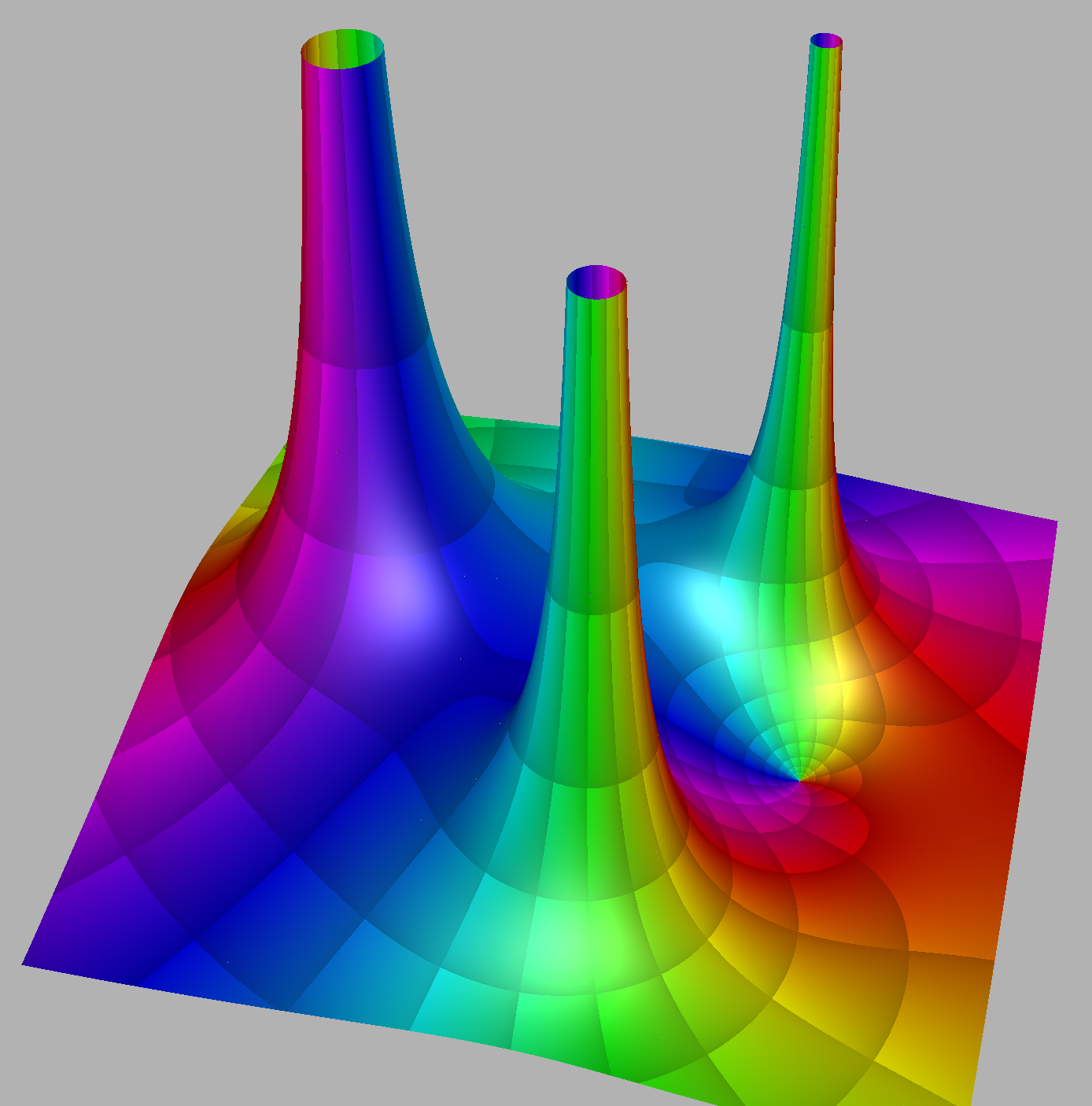}
    \end{adjustbox}

    \caption{\emph{CindyGL} rendered visualizations of the function $f:\mathbb{C}\to\mathbb{C}, z\mapsto \frac{z-1}{z^3+1-i}$: (a) Phase portrait of $f$, (b) truncated Taylor expansion of $f$ at $z_0 = \tfrac{1}{2}$, (c) stereographic projection of $f$ on the Riemann sphere making the double root at $\infty$ visible, (d) analytic landscape of $f$.}
    \label{fig:education}
  \end{center}
\end{figure}

Advanced tasks were assigned to two university students as part of their bachelor's theses, with \emph{CindyGL} provided as a tool. Both students started experimenting with \emph{CindyGL}, drew their mathematical conclusions, and, furthermore, created interactive GPU-based teaching material without previous knowledge of shader programming.
 In his bachelor thesis, \cite{konnerth2017} simulated several concepts of complex analysis with phase portraits of partial sums of the Taylor series (see Fig. \ref{fig:education}b), a ray-casted Riemann-sphere (see Fig. \ref{fig:education}c), and renderings of analytic landscapes (see Fig. \ref{fig:education}d).
 In the other bachelor's thesis, the rolling shutter effect was simulated on the GPU.

\section{Conclusion and Outlook}\label{conclusion-and-outlook}
In this article, we worked out a concept to narrow the gap between CPU and GPU programming for fast prototyping environments such as dynamic geometry software. The key component is the automatic transcompilation of an untyped scripting language to a shader language. The parts of the code that are candidates for parallelization are automatically detected.

In our approach, relevant code is detected through a recursive analysis of the syntax tree, while a fixed-point algorithm determines the types.

We did a sample implementation as a working proof of concept through \emph{CindyGL}. The core component of \emph{CindyGL} is a transcompiler, which can translate the \emph{CindyJS} inherent scripting language \emph{CindyScript} to the OpenGL Shading Language (\emph{GLSL}).

The process could be enhanced by detecting more complicated patterns that are suitable for parallelization on the GPU. We also have shown that more computationally intensive high-level operations in dynamic geometry software such as raycasting or physic simulations can be computed on the GPU within a DGS.

\bibliography{mybib}

\end{document}